\documentclass[a4paper]{article}
\usepackage[utf8]{inputenc}
\usepackage[english]{babel} 
\usepackage{amsmath,amssymb,amsthm,bm} 
\usepackage[margin=2cm]{geometry}
\usepackage{esint}
\usepackage[colorlinks]{hyperref}
\usepackage{enumerate}
\usepackage{framed,comment}
\usepackage{upgreek}
\usepackage{pgfplots}
\usepackage{float}
\usepackage{tikz,tikz-cd}
\usepackage{rotating}
\usepackage{appendix}

\newtheorem{theorem}{Theorem}[section]

\newtheorem{example}[theorem]{Example}

\newtheorem{proposition}[theorem]{Proposition}
\newtheorem{remark}[theorem]{Remark}

\usepackage{parskip}
\numberwithin{equation}{section}

\newcommand{\rmd}{{\color{red}{\rm d}}}
\newcommand{\scp}[2]{{\big\langle {#1}\, , \, {#2}\big\rangle}}
\newcommand{\Scp}[2]{{\Big\langle {#1}\, , \, {#2}\Big\rangle}}
\newcommand{\SCP}[2]{{\left\langle {#1}\, , \, {#2}\right\rangle}}

\newcommand{\qp}[2]{{\left( {#1}\, , \, {#2}\right)}}
\newcommand{\qv}[2]{{\left[ {#1}\, , \, {#2}\right]}}

\newcommand{\mb}[1]{\mbox{\boldmath{$#1$}}}
\newcommand{\bs}[1]{\boldsymbol{#1}}

\newcommand{\E}[1]{\mathbb{E}\left[{#1}\right]}

\def\contract{\makebox[1.2em][c]{\mbox{\rule{.6em}
			{.01truein}\rule{.01truein}{.6em}}}}

\def\p{\partial}

\def\GRAD{{\bm \nabla}}

\def\DIV{{\bm \nabla \cdot}}

\def\CROSS{{\bm \times}}
\def\div{\hbox{div}}
\def\curl{\hbox{curl}}

\def\ep{{\epsilon}}

\def\bx{{\bf x}}

\def\bfm{{\bf m}}
\def\bR{{\bf R}}
\def\bS{{\bf S}}
\def\bu{{\bf u}}

\def\bv{{\bf v}}

\def\bd{{\bf d}}

\def\bff{{\bf f}}

\def\bsigma{{\bs \sigma}}

\def\ad{{{\rm ad}}}
\def\Ad{{{\rm Ad}}}

\def\p{{\partial}}

\def\rmd{{\color{red}{\rm d}}}

\pgfplotsset{compat=1.14}

\title{Stochastic effects of waves on currents in the ocean mixed layer}
\author{Darryl D. Holm\quad\hbox{and}\quad Ruiao Hu \\ \smallskip\large
	Department of Mathematics\\ 
	Imperial College London SW7 2AZ, UK
	\\ \bigskip\normalsize
	d.holm@ic.ac.uk and ruiao.hu15@ic.ac.uk} 
\date{ }

\begin{document}
	
\maketitle

\begin{abstract}


	This paper introduces an energy-preserving stochastic model for studying wave effects on currents in the ocean mixing layer. The model is called stochastic forcing by Lie transport (SFLT). The SFLT model is derived here from a stochastic constrained variational principle, so it has a Kelvin circulation theorem. The examples of SFLT given here treat 3D Euler fluid flow, rotating shallow water dynamics and the Euler-Boussinesq equations. In each example, one sees the effect of stochastic Stokes drift and material entrainment in the generation of fluid circulation. We also present an Eulerian-averaged SFLT model (EA SFLT), based on decomposing the Eulerian solutions of the energy-conserving SFLT model into sums of their expectations and fluctuations. 
	
\end{abstract}




\section{Introduction}





\paragraph{Wave effects on currents (WEC).}
In studies of the ocean mixing layer (OML) the problem of wave effects on currents (WEC) arises. For example, surface gravity waves can drive Langmuir circulations which have important influences on near-surface currents \cite{LC-pics}. Langmuir circulations are horizontally oriented pairs of oppositely circulating vortex tubes aligned generally  along the direction of the wind, as reviewed, e.g., in \cite{Leibovich1983,Garrett1996,Thorpe2004,{F-K2018}}. Because they represent arrays of organised fluid transport, Langmuir circulation patterns can produce vertical transport which strongly entrains sediment and detritus into the OML from both above and below. This means they can have a strong effects, for example, on the dispersion of oil spills in a shallow sea, \cite{Thorpe2001}. 

Most previous studies of turbulence in the OML have been performed in the context of wave-averaged dynamics. For a recent review, see \cite{AY2020}. The underlying assumption is that the surface gravity waves represent the fastest component in the system and are only weakly modulated by the other components (turbulence and currents). Averaging the stratified Euler equations in three spatial dimensions over a time scale longer than the wave period produces a modified set of equations, known as the CL equations, after Craik and Leibovich, \cite{CL1976}. 

The temporal-averaging basis of the CL equations introduces Stokes-drift effects which represent additional wave-averaged forces and material advection terms that emerge from multi-scale asymptotic theories \cite{CL1976, McWRL2004}. Remarkably, the ideal CL equations preserve many of the properties of the original stratified Euler equations. For example, the ideal CL equations conserve energy and have a Hamiltonian formulation \cite{HolmCL1996}. 

The aim of the present paper is to develop a stochastic theory of WEC which encompasses the Craik-Leibovich theory and still preserves the original energy. Indeed, in deriving the stochastic theory, we will take an approach which can be adapted to select whichever primary conservation laws of the deterministic theory are desired. Specifically, our approach will use a stochastic version of a method from classical mechanics known as the Reduced Lagrange-d'Alembert-Pontryagin (RLDP) formulation of constrained dynamics which is reduced by a symmetry of the Lagrangian in Hamilton's principle. Stochastic applications of the RLDP formulation are discussed, e.g., in \cite{B-RO2009,GH18a}. 

\paragraph{RLDP result for the stochastic Craik-Leibovich theory for the Euler-Boussinesq equations. }
Specific results of the theoretical developments in this paper can be assessed by simply examining the example of applying the RLDP approach to the ideal CL equations themselves in section \ref{subsec: EB}. 

One finds in this example that the energy-preserving stochastic EB equations in \eqref{eq:SEB eq} produce a stochastic contribution to the vortex force in the Craik-Leibovich equations \cite{CL1976} whose deterministic formulation with Hamilton's principle is given in \cite{HolmCL1996}. Namely, they reduce as follows, where $\rmd \bu$ denotes the Stratonovich stochastic time differential of Eulerian fluid velocity $\bu$, the Stokes drift velocity is $\bu^S(\bx)$, the Coriolis parameter is $2 \bs{\Omega}={\rm curl}\bR(\bx)$, the pressure is $p$, the volume element is $D$, the buoyancy is $b$ with gravitational constant $g$, and the stochastic spatial modes obtained via data calibration are denoted $\bff^k(\bx)$
	\begin{align}
	\begin{split}    
	\rmd \bu - \bu\times {\rm curl}\big(\bu - \bu^S(\bx) + \bR(\bx)\big) dt
	&= -\,\nabla \Big( p + \frac12 |\bu|^2 + \bu\cdot\bu^S \Big)dt  - gb\, \mathbf{\hat{z}} \, dt
	\\&\qquad + \sum_{k>0} \Big( \bu\times {\rm curl}\,\bff^k 
	- \,\nabla \big(\bu\cdot \bff^k\big)\Big) \circ dW^k_t
	\, ,   \\
	\partial_t D + {\rm div}(D\bu) &= 0\, ,   \quad \hbox{with}\quad D=1
	\, ,   \\
	\partial_t b + \bu\cdot \nabla b &= 0\,, 
	\end{split}   
	\label{SCL-eqns}
	\end{align}
	where one interprets the semimartingale $\rmd \bu^S = \bu^S(\bx)\,dt + \sum_k \bff^k(\bx)  \circ dW^k_t $ as a stochastic augmentation of the usual steady prescribed Stokes drift velocity and one recovers the deterministic Craik-Leibovich equations when the stochastic terms proportional to $\circ dW^k_t$ are absent. 
	
	The energy-preserving stochastic EB equations in \eqref{eq:SEB eq} imply the following equation for potential vorticity density, defined by $q := (\curl \bfm)\cdot \GRAD b  = (\bs{\omega} + 2 \bs{\Omega})\cdot \GRAD b$, where $\bfm=\bu+\bR$ for $D=1$. Namely,
\begin{align}
\rmd q + \bu\cdot\GRAD q\,dt =: -\,\sum_k {\rm div} \mathbf{J}^k \circ dW^k_t  
\,,
\label{eqn: PV-StochEB-intro}
\end{align}
in which the ``$J$-fluxes of PV'' on the right-hand side are discussed, e.g., in \cite{HM1987, Marshall&Nurser1992}.   See also \cite{BodnerFK2020} for LES turbulence interpretations of these fluxes.

In summary, while the energy-preserving stochastically-augmented CL vortex force and entrainment effects in this example can locally create stochastic Langmuir circulations, the total volume-integrated potential vorticity $Q = \int_{\mathcal{D}}q \,dV$ will be preserved for appropriate boundary conditions.

\subsection{Motivating question and main results of the paper} 

The present paper addresses the geometric interplay between energy and circulation, when stochasticity is introduced into fluid dynamics, in both the incompressible flow of an ideal Euler fluid and in the flows of ideal fluids with advected quantities. This is a burgeoning area of research in fluid dynamics. For recent introductory surveys of stochastic fluid dynamics with applications, see, e.g., \cite{Cruzeiro2020,FP2020,GH2020}.

The question underlying the present work is, ``What types of noise perturbations can be added to fluid dynamics which will preserve its fundamental properties of energy conservation, Kelvin-Noether circulation dynamics and conserved properties resulting from invariance under Lagrangian particle relabelling?'' Since these properties all arise from the geometric structure of fluid dynamics, the noise perturbations we consider will be introduced in a geometrical framework. 

The RLDP formulation of fluid dynamics leads to a constrained variational principle for stochastic ideal fluid dynamics through which noise may be introduced as a prescribed stochastic force. By \emph{choosing} an appropriate form of the stochastic force, the RLDP formulation can be designed to preserve whatever conservation law one may desire among those of the deterministic ideal fluid equations in any number of dimensions. 

For applications in fluid dynamics, the RLDP formulation results in the procedure mentioned above, called Stochastic Forcing by Lie Transport (SFLT). This is constructed by choosing the stochastic forces to Lie derivative of the transport velocity vector field acting on some prescribed one-form density which conserves the ideal fluid energy. Potential applications include a stochastic version of the Craik-Leibovich (CL) vortex force which presumably could generate stochastic Langmuir circulation. This result is introduced for the Euler fluid equations for 3D incompressible flow in Example \ref{Prop: Energy-conserved} and Proposition \ref{Cor: VortexForce}. The corresponding theory for ideal flows in general is also treated in section \ref{sec: SDP SFLT-erg} where stochastic advected quantities are introduced. Using RLDP, the existence of Kelvin-Noether theorem is automatic due to the variational nature of the formulation. When the advection of the fluid density is assumed, the Kelvin circulation theorem becomes a simple corollary of the Kelvin-Noether theorem as proven in Theorem \ref{thm:SFLT KN}.

\paragraph{SFLT has dual design capabilities.}
Besides adding stochastic \emph{forces} which drive the fluid motion equation, the SFLT framework can also be designed to distribute stochastic \emph{sources} in the advective transport equations, as discussed in section \ref{sec: SDP SFLT-erg}. The stochastic sources distributed in the advective transport equations by SFLT can be designed to model, for example, stochastic changes in the material properties of inertial fluid parcels which may be embedded in the flow. In particular, one can use SFLT to model the stochastic dynamics of a mixture of heavier and lighter parcels whose fluid paths deviate from passive tracers, which are carried by the drift flow velocity.  In this case, the mass density would be changing stochastically in the material frame of the flow drift velocity, because of entrainment or detrainment of sediments. One can observe the entrainment or detrainment of sediments in Langmuir circulations, see for example, \cite{LC-pics}. One can also imagine using SFLT to model the transport of a stochastically evolving, spatially distributed, algae bloom which has a distribution of shapes, so it is only partially embedded in a flow around an obstacle, such as an island in the ocean \cite{globcurrent,Messie2020}. 

Thus, SFLT has dual design capabilities. It can transport stochastic \emph{sources} in the material frame, and it can also impose stochastic \emph{non-inertial forces} arising from stochastic changes of the frame of motion, such as the Craik-Leibovich \emph{vortex force} \cite{CL1976,HolmCL1996}. The dual design capabilities of SFLT could potentially lead to a variety of important applications. For example, the use of SFLT for modelling entrainment and detrainment of various materials into flows in the ocean mixed layer (OML) can in principle be instrumental in modelling some important components of natural processes such as gas and nutrient exchanges. Besides its potential importance in modelling natural processes in the transport of materials such as sediment or algae blooms  in the OML, the SFLT stochastic modelling framework may also find a role in data assimilation for predicting the transport of pollution such as oil droplets, microplastics, etc. 

The energy preserving SFLT framework is extended to the Eulerian Averaged SFLT (EA SFLT) framework by applying an \emph{Eulerian Average} on the Eulerian quantities in the equations. These systems are \emph{non-local} in probability space in the sense that the expected momentum density and expected advected quantities are assumed to replace the drift momentum density and advected quantities respectively. These equations retain the fundamental properties of fluid dynamics and in addition, they have potential uses in climate change science since the quantities of interests have a clear sense of expectation and fluctuation, and the expected fluid motion is deterministic. For quadratic Hamiltonians, we find closed-form equations for the dynamics of the expectations and fluctuations of the Eulerian fluid variables. Total energy conservation enables us to show that, over time, the energy of the expected quantities is converted into the energy of the fluctuations, while the sum remains the same. 

To illustrate the dual design capability of the proposed SLFT framework and the EA SFLT framework, explicit applications for the motion of fluids under gravity with SFLT noise will be given in the last part of the paper. Section \ref{subsec: HeavyTop} deals with the heavy top, which is a finite degree-of-freedom subsystem of fluid motion \cite{Holm1986}; in section \ref{subsec: RSW} for rotating shallow water dynamics; and in section \ref{subsec: EB} for the stochastic Euler-Boussinesq equations. 

\paragraph{List of abbreviations}
\begin{itemize}
    \item
    Geophysical Fluid Dynamics (GFD)
    \item
    Wave Effects on Currents (WEC)
    \item 
    Reduced Lagrange d'Alembert Principle (RLDP)
    \item
    Craik-Leibovich (CL)
    \item
    Generalised Lagrangian Mean (GLM)
    \item 
    Stochastic Forcing by Lie Transport (SFLT)
    \item 
    Eulerian Averaged Stochastic Forcing by Lie Transport (EA SFLT)
    \item
    Stochastic Advection by Lie Transport (SALT)
    \item
    Lagrangian Averaged Stochastic Advection by Lie Transport (LA SALT)
    \item
    Stochastic Forced Euler-Poincar\'e (SFEP)
    \item
    Stochastic Forced Lie-Poisson (SFLP)
    \item 
    Rotating Shallow Water (RSW)
    \item 
    Euler-Boussinesq (EB)
\end{itemize}

\subsection{Plan of the paper} 



In section \ref{sec: SFLT noise} we introduce SFLT noise in the RLDP form and show that this formulation is flexible enough to accommodate a variety of different types of noise perturbations with potential applications to stochastic fluid dynamics.
In section \ref{sec: EA SFLT} we introduce the EA SFLT framework as an extension of the SFLT framework. We show that this modification to SFLT have potential applications to climate change science.
In section \ref{sec: Examples} we present several examples of conservative noise types for semidirect-product coadjoint motion. These include the finite-dimensional case of the heavy top (which is a gyroscopic analog for collective motion of a stratified fluid \cite{Holm1986}), and the infinite-dimensional fluid cases of rotating shallow water (RSW) dynamics in 2D and Euler-Boussinesq (EB) dynamics in 3D.
The paper emphasises the utility of the RLDP formulation in introducing a variety of stochastic perturbations which may be chosen to preserve the properties of the deterministic solutions of the fluid dynamics equations that are of most concern and value to the modeller. In particular, RLDP admits a unified variational formulation which combines the variational principle used in \cite{Holm2015} and SFLT, as explained in remark \ref{remark: Unified approach}.

Three appendices have been provided. These appendices are meant to supply supporting details without interfering with the main flow of the paper. They contain further discussions of the following topics related to the main part of the paper: 
\ref{app: semidirect ad^*} Coadjoint operator of semidirect-product Lie-Poisson brackets;
\ref{app: SALT} Stochastic advection by Lie transport (SALT); 
\ref{app: Ito form} It\^o form of the SFLP equation;

\paragraph{Slow + Fast decompositions of fluid flows.}
In modelling geophysical fluid dynamics (GFD) in ocean, atmosphere or climate science, one tends to focus on balanced solution states which are near certain observed equilibrium \cite{McWilliams2003,F-K2018}. These equilibria include hydrostatic and geostrophic balanced states, for example, in both the ocean and the atmosphere. Upsetting these balances can introduce both fast and slow temporal behaviour. The response depends on the range of the frequencies in the spectrum of excitations of the system away from equilibrium under the perturbations. When a separation in time scales exists in the response of the system to perturbations of one of its equilibria, then one may propose to average over the high frequency response and retain the remaining slow dynamics which remains near the equilibrium. This happens, for example, in the quasigeostrophic response to disturbances of geostrophic equilibria in the 2D rotating shallow water equations. Averaging over the high frequencies also often produces a slow \emph{ponderomotive force}, due for example to a slowly varying envelope which modulates the high frequency response.%
\footnote{In taking these averages over the fast behaviour of the dynamics one must also deal properly with any resonances which would occur. However, the effects of resonances will be neglected in our discussion here.}  In most situations in GFD, though, the solution only stays near the low-frequency \emph{slow manifold} for a rather finite time before developing a high-frequency response, See, e.g., Lorenz \cite{Lorenz1986,Lorenz1987,Lorenz1992}. In practice, for example in numerical weather prediction, the emergence of the high-frequency disturbances of the devoutly wished slow manifold introduces undeniable uncertainty which historically has often been handled by some sort of intervention, such as nonlinear re-initialisation \cite{Leith1980}. Sometimes the effects of the emergence of high frequencies can be treated to advantage in computational simulations. For example, the stochastic back-scatter approach of Leith \cite{Leith1990} is commonly used in computational simulations to feed energy from the burgeoning unstable development of high-wavenumber excitations into large-scale coherent structures at low frequencies, so as to enhance the formation of eddies in ocean flows,  \cite{Berloff2005}.

\paragraph{Hamilton's principle and Kelvin's circulation theorem.} An opportunity for further theoretical understanding of the interactions of disparate scales in GFD (and, for example, in astrophysics and planetary physics \cite{Pouquet2009}) arises when the non-dissipative part of the fluidic system dynamics under consideration in a domain $\mathcal{D}$ can be derived from Hamilton's principle, $\delta S=0$, for an action time-integral given by $S = \int \ell(u,a)dt$, where the fluid Lagrangian $\ell(u,a)$ depends on Eulerian fluid variables comprising the fluid velocity vector field $u\in \mathfrak{X}(\mathcal{D})$ and some set of advected quantities, $a\in V^*(\mathcal{D})$, dual in $L^2$ pairing to a vector space, or tensor space, $V$, defined over the domain $\mathcal{D}$ with appropriate boundary conditions. This approach leads directly to a Kelvin-Noether circulation theorem arising from the symmetry of the Lagrangian $\ell(u,a)$, written in terms of Eulerian fluid variables, under transformations of the initial material labels which preserve the initial conditions of the advected quantities along the particle trajectories in the flow \cite{HMR1998}. In this case, the averaging over high frequencies in the solution may be applied by substituting a WKB (slowly varying complex amplitude times a fast but slowly varying phase) decomposition of the fluid parcel trajectory into the Lagrangian, then averaging over the fast phase before taking variations.  A prominent example of this approach for applications in GFD is the General Lagrangian Mean (GLM) phase-averaged description of the interaction of fluctuations with a mean flow introduced in \cite{AM1978} and developed further in \cite{GjHo1996,Holm2002a,Holm2002,XV2015,AY2020}. Many of the ideas underlying GLM are also {\it standard} in the stability analysis of fluid equilibria in the Lagrangian picture. See, e.g., the classic stability analysis papers of \cite{Bernstein1958, FR1960, Newcomb1962, Hayes1970}. 

\paragraph{Stochastic variational principles.}
The present paper discusses yet another opportunity for introducing a slow-fast decomposition for the sake of further understanding of GFD. This opportunity arises when the separation in time scales in the symmetry-reduced Lagrangian $\ell(u,a)$ for Eulerian fluid variables can be posed as the decomposition of the fluid velocity $u$ into the sum of a deterministic drift velocity modelling the computationally resolvable scales and a stochastic velocity vector field modelling the correlates at the resolvable scales of the computationally unresolvable \emph{sub-grid} scales of motion. 
A stochastic variational principle using the slow-fast decomposition of the Lagrangian flow map led to the derivation of stochastic Euler--Poincar\'{e} (SEP) equations in \cite{Holm2015}. In \cite{Holm2015}, the fast motion of the Lagrangian trajectory is represented by a stochastic process, whose correlate statistics are to be calibrated from data as in \cite{CCHOS18,CCHOS18a}. The  stochastic decomposition proposed in \cite{Holm2015} was later derived using multi-time homogenization by Cotter et al.  \cite{CGH17}. Transport of fluid properties along an ensemble of these stochastic Lagrangian trajectories is called Stochastic Advection by Lie Transport (SALT). The well-posedness of the Euler fluid version of the Euler--Poincar\'{e} SALT equations in three dimensions was established in \cite{CFH19} for initial conditions in appropriate Sobolev spaces. The mathematical framework of semimartingale-driven stochastic variational principles was established in \cite{SC2020}.

\paragraph{The geometric interplay between energy and circulation in ideal fluid dynamics.}
Fluid dynamics transforms energy into circulation. As it turns out, this transformation is quite geometric. In particular, the solutions of Euler's fluid equations for ideal incompressible flow describe geodesic curves parametrised by time on the manifold of volume-preserving diffeomorphisms (smooth invertible maps). These geodesic curves are defined with respect to the metric provided by the fluid's kinetic energy, defined on the smooth, divergence-free, velocity vector fields which comprise the tangent space of the volume-preserving diffeomorphisms ${\rm SDiff}(\mathcal{D})$ acting on the flow domain $\mathcal{D}$. This 1966 result of V. I. Arnold   \cite{Arnold1966} introduced a fundamentally new geometric way of understanding energy and circulation in fluid dynamics. The incompressible Euler fluid case in \cite{Arnold1966} possesses the well-known conservation laws of energy and circulation, both of which arise via Noether's theorem from symmetries of the Lagrangian in Hamilton's variational principle under the right action of ${\rm SDiff}(\mathcal{D})$. Later, Arnold  \cite{Arnold1974} noticed the topological nature of another conservation law for the Euler fluid equations which is known as helicity. In particular, the conserved  helicity measures the topological linkage number of the vorticity field lines in Euler fluid dynamics. 

Some of this geometric interplay between energy and circulation in ideal fluid dynamics already shows up in Kelvin's circulation theorem \cite{Kelvin1869} for ideal Euler fluids, which emerges as the Kelvin-Noether theorem from right-invariance (relabelling symmetry) of the Lagrangian in Hamilton's principle when written in terms of the Eulerian representation, \cite{HMR1998}. 
In physical fluids, the particle-relabelling symmetry of the fluid Lagrangian in the Eulerian representation is broken from the full diffeomorphism group $G={\rm Diff}(\mathcal{D})$ to its subgroup $G_a={\rm Diff}(\mathcal{D})|_{a_0}$ which leaves invariant the initial conditions $a_0$ for advected fluid variables, denoted $a$, such as the mass density and thermodynamic properties. The Lagrangian histories $x_t=g_tx_0$ evolve as $\dot{x}_t=\dot{g}_tx_0 = u(g_t x_0,t)=u(x_t ,t)$, where the Eulerian velocity vector field given by $u := \dot{g}_t g_t^{-1}$ is right-invariant under the particle-relabelling $x_0\to y_0=h_0x_0$ for any fixed $h_0\in G$. 

In this geometrical setting for fluids in the Eulerian representation, the Legendre transform maps the Lagrangian variational formulation to the Hamiltonian formulation in which a Lie-Poisson bracket governs the motion generated by the Hamiltonian. This Lie-Poisson bracket is defined on the dual space of the Lie algebra of divergence-free vector fields $\mathfrak{X}_{\rm div}(\mathcal{D})$ for Euler fluids. For ideal fluids with advected quantities, the Lie-Poisson bracket is defined on the dual space of the SDP Lie algebra $\mathfrak{X}(\mathcal{D})\circledS V(\mathcal{D})$. 
As discussed in \cite{HMR1998}, the fluid motion in each case represents the coadjoint action of the corresponding Lie algebra on its dual space.
For more details in the present context, refer to appendices \ref{app: semidirect ad^*} and \ref{app: SALT}, particularly  \ref{app: KIW}.


%

\paragraph{Geometric formulation of advective transport.} The action of the full diffeomorphism group $G={\rm Diff}(\mathcal{D})$ on its \emph{order parameter} variables $a\in V(\mathcal{D})$ represents fluid advection. Advection occurs by \emph{push-forward} of functions on $V(\mathcal{D})$ (by right action by the inverse of the Lagrange-to-Euler map). This means the time-dependence of an advected quantity is given by the push-forward relation for composition of functions; namely,
\[
a(t) = g_{t\,*}a_0 := a_0g_t^{-1} 
\,.\] 
Thus, an advected quantity $a(t)$ evolves by the \emph{Lie chain rule} 
\[
\p_t a  =  - \mathcal{L}_{\dot{g}_t g_t^{-1} }a_t = - \mathcal{L}_{u_t}a
\,,
\]
in which $ \mathcal{L}_{u_t}a$ denotes the \emph{Lie derivative} of the advected quantity $a\in V(\mathcal{D})$ by the time-dependent Eulerian velocity vector field, $u_t:=\dot{g}_t g_t^{-1}\in\mathfrak{X}(\mathcal{D})$. For more details about performing this type of calculation for \emph{Lie transport}, see appendix \ref{app: semidirect ad^*}. 
The applications of these ideas in developing the SALT approach are discussed in appendix \ref{app: SALT}.

\section{Stochastic forcing by Lie transport (SFLT)}\label{sec: SFLT noise}

After briefly surveying in section \ref{subsec-RDLP} a few of the available capabilities arising from the RLDP principle for the sake of other potential research directions, we will continue in section \ref{sec: SDP SFLT-erg} toward our primary objective to develop the energy preserving SFLT theory for applications to semidirect-product fluid motion and material entrainment in the remainder of the paper.

\subsection{Reduced Lagrange d'Alembert Pontryagin (RLDP) Principle}\label{subsec-RDLP}

\paragraph{Stochastic non-inertial frames.}
Newton's law of motion in a non-inertial frame redefines momentum so as to introduce an additional force. The canonical example is the Coriolis force, which arises from redefining the momentum in a rotating frame of motion in terms of the velocity as viewed from an inertial frame. The CL vortex force arises in the fluid momentum equation by the addition of Stokes drift velocity to the frame of motion of the fluid momentum variable $m \in \mathfrak{X}^*$. In the present setting, it is natural to consider a stochastic addition to the momentum which corresponds to a stochastic change of referemce frame. 

One variational principle which permits such additional forces is the \emph{reduced Lagrange d'Alembert Pontryagin} (RLDP) principle. The general construction of fluid dynamics using the RLDP formulation is as follows \cite{GH18a}.\footnote{See \cite{B-RO2009} for the corresponding result for the finite-dimensional Euler-Lagrange equation in the absence of symmetry.} 
Consider the Lie group $G = {\rm Diff}(\mathcal{D})$ with an associated Lie algebra $\mathfrak{X}$, the RLDP principle for reduced Lagrangian $\ell:\mathfrak{X}\rightarrow \mathbb{R}$ and external force $F \in \mathfrak{X}^*$ is given by
\begin{align}
\delta \int_b^a \ell(u) + \SCP{m}{\dot{g}\,g^{-1}- u}\,dt - \int_b^a \SCP{F}{\delta g\,g^{-1}}\,dt = 0\,,
\label{eqn:RLDP}
\end{align}
where $g\in G$, the variations $\delta g, \delta u, \delta m$ are arbitrary with $\delta g$ vanishes at the end points $t=a,b$. The stationary condition \eqref{eqn:RLDP} yields the forced Euler-Poincar\'e equation with force $F$, as a 1-form density equation,
\begin{align}
\frac{\partial}{\partial t}\frac{\delta \ell}{\delta u} + \ad^*_u \frac{\delta \ell}{\delta u} + F = 0\,.
\label{Det-EPeqn}
\end{align}

Adding stochasticity to the RLDP principle with external forces can be done, as follows. For reduced Lagrangian $\ell:\mathfrak{X}\rightarrow \mathbb{R}$ and the set of external forces $F_i \in \mathfrak{X}^*$, stochastic RLDP is given by
\begin{align}
\delta \int_b^a \ell(u) + \SCP{m}{\rmd g\,g^{-1}- u}\,dt - \sum_i \int_b^a \SCP{F_i}{\eta}\circ dW_t^i = 0\,, 
\label{SRLDP action on G}
\end{align}
where $g\in G$, the variations $\delta g, \delta u, \delta m$ are arbitrary, $\delta g$ vanishes at the end points, $t=a,b$, and $\eta:=\delta g\,g^{-1}$. The stationary condition of the variational principle \eqref{SRLDP action on G} yields the following \emph{stochastically forced Euler-Poincar\'e} (SFEP) equation,
\begin{align}
\rmd \frac{\delta \ell}{\delta u} + \ad^*_u \frac{\delta \ell}{\delta u}\,dt + \sum_i F_i \circ dW^i_t = 0\,.
\label{Stoch-EPeqn}
\end{align}

\paragraph{Lie-Poisson Hamiltonian formulation.} Upon passing to the Hamiltonian side via the Legendre transform $\ell(u) = \SCP{m}{u} - h(m)$, where $m = \frac{\delta \ell}{\delta u}$ and $h$ is the reduced Hamiltonian, one finds the \emph{reduced Hamilton-d'Alembert Pontryagin phase space principle} given by
\begin{align}
\delta \int_b^a \left[\SCP{m}{\rmd g\,g^{-1}} - h(m)\right]\,dt - \sum_i \int_b^a \SCP{F_i}{\eta}\circ dW_t^i = 0\,.
\label{SRLDP-Ham}
\end{align}
As before, $g\in G$, the variations $\delta g, \delta m$ are arbitrary and $\delta g$ vanishes at the endpoints in time, $t=a,b$. The resulting \emph{stochastically forced Lie-Poisson} (SFLP) equation is
\begin{align}
\rmd m + \ad^*_u m\,dt + \sum_i F_i \circ dW^i_t = 0\,, \quad \text{where } u 
= \frac{\delta h}{\delta m}\,.
\label{SFLP-general-noise}
\end{align}
Within the construction of class of SFLP equations, choices of the external forces $F_i$ exists. With specific choices of the external forces $F_i$, the resulting equations can be either energy-preserving, or Casimir preserving as shown by subsequent examples.
\begin{example}[Energy-preserving SFLP equations] \label{Prop: Energy-conserved}
    Let $F_i = \ad^*_u f_i$ where $f_i \in \mathfrak{X}^*$, then the SFLP equation \eqref{SFLP-general-noise} becomes 
	\begin{align}
	\rmd m + \ad^*_u m\,dt + \sum_i \ad^*_u f_i \circ dW^i_t = 0\,, 
	\quad \text{where } u = \frac{\delta h}{\delta m}\,, 
	\label{SAPS-LP-eq}
	\end{align}
	in agreement with \cite{DrivasHolm2019}. 
	Energy preservation may now be immediately verified, since 
	\begin{align}
	\rmd h = \SCP{\rmd m}{\frac{\delta h}{\delta m}} = \SCP{-\ad^*_u m\,dt - \sum_i \ad^*_u f_i \circ dW^i_t}{u} = 0\,,
	\end{align}
	in which the last equality follows because of the anti-symmetry of the commutator, ${\rm ad}_u u = - [u ,u ] = 0$. However, the Casimirs are no longer conserved, because the Lie-Poisson operator in equation \eqref{SAPS-LP-eq} has been changed by the addition of noise. For completeness, the It\'o form of the equation \eqref{SAPS-LP-eq} is presented in Appendix \ref{app: Ito form}.
\end{example}

\begin{proposition}[Vortex force]\label{Cor: VortexForce}
    The energy preserving SFLP equation \eqref{SAPS-LP-eq} contains stochastic vortex forces.
\end{proposition}
\begin{proof}
    The stochastic term in \eqref{SAPS-LP-eq} can be expressed as
	\begin{align}
	\sum_i \ad^*_u f_i \circ dW^i_t 
	= \sum_i \Big(\bs{u}\times {\rm curl} \bs{f}_i - \nabla (\bs{u}\cdot \bs{f}_i)
	 -\bs{f}_i \div{\,\bu}\Big)\cdot d\bs{x} 
	\circ dW^i_t \,,
	\label{Vortex-force}
	\end{align}
	where $\bs{u}$ and $\bs{f}_i$ are the coefficients of the vector field $u$ and 1-form densities $f_i$. i.e. $u = \bs{u}\cdot \frac{\partial}{\partial \bs{x}}$ and $f_i= \bs{f}_i\cdot d\bs{x} \,d^3x$. Thus, a stochastic version of the CL vortex force $\sum_i \bs{u}\times {\rm curl} \bs{f}_i \circ dW^i_t $ with a stochastic contribution to the pressure $\sum_i  \GRAD(\bs{u}\cdot\bs{f}_i)\circ dW^i_t$ emerges in the energy-preserving form of the RLDP equations in \eqref{SAPS-LP-eq}, cf. \cite{CL1976}. Note that the term $\bs{f}_i\div \bs{u}$ vanishes for incompressible flow.
\end{proof}
For a specific choice of Hamiltonian, the forces presented in example \ref{Prop: Energy-conserved} may not be the only forces that conserve energy. However, the $\ad_u f_i^*$ form of the noise is geometric, so the energy preserving property does apply to all fluid systems that have a variational formulation. Thus, it is possible to construct a class of energy preserving stochastic systems using the geometric formulation. As discussed further in section \ref{sec: SDP SFLT-erg}, the noise terms in this class of energy preserving stochastic systems appear the form of a `frozen' (constant coefficient) Lie-Poisson bracket, whose properties are discussed, e.g., in Appendix B of \cite{HMRW1985}.
\begin{example}[Casimir-preserving SFLP equations]\label{Prop: Casimir-conserved}
	Given a Casimir function $C(m)$, the Casimir preserving forces satisfy $F_i = \ad^*_{\frac{\delta C}{\delta m}}f_i$ where $f_i$ are arbitrary. Then 
	\begin{align*}
	\rmd C = \SCP{\rmd m}{\frac{\delta C}{\delta m}} = \SCP{-\ad^*_u m\,dt - \sum_i \ad^*_{\frac{\delta C}{\delta m}} f_i \circ dW^i_t}{{\frac{\delta C}{\delta m}}} = 0,
	\end{align*}
	\emph{provided} both $\SCP{\ad^*_u m}{{\frac{\delta C}{\delta m}}} = 0$ and $\SCP{\ad^*_{\frac{\delta C}{\delta m}} f_i}{{\frac{\delta C}{\delta m}}} = 0$, because of the degeneracy of the LP bracket and the anti-symmetry of the commutator, respectively. One concludes that the choice of external forces $F_i = \ad^*_\frac{\delta C}{\delta m} f_i$ cannot in general preserve all of the Casimirs seen in the unperturbed LP equation, unlike the case with SALT, where \emph{all of the Casimirs} are preserved. Note that the energy are no longer conserved, since
	\begin{align*}
	    \SCP{\rmd m}{\frac{\delta h}{\delta m}} = \SCP{-\ad^*_{\frac{\delta h}{\delta m}} m\,dt - \sum_i \ad^*_{\frac{\delta C}{\delta m}} f_i \circ dW^i_t}{\frac{\delta h}{\delta m}} = \SCP{- \sum_i \ad^*_{\frac{\delta C}{\delta m}} f_i \circ dW^i_t}{\frac{\delta h}{\delta m}}\,,
	\end{align*}
	does not vanish trivially.
\end{example}\bigskip

\begin{proposition}[Helicity preservation for a stochastic 3D Euler fluid]\label{Prop: Helicity-conserved}$\,$
	
	The deterministic 3D Euler fluid equations preserve the helicity,
	\[
	C(m) = \langle m,\bd m\rangle=\frac12\int \bs{m}\cdot {\rm curl}\bs{m}\,d^3x\,,
	\] 
	where $m=\bs{m}\cdot d \bs{x}$ is the circulation 1-form.
	The SFLT model will preserve the helicity $C(m)$ when the constraint force is taken to be 
	\[
	F_i = \ad^*_{\frac{\delta C}{\delta m}}f_i \in\mathfrak{X}^*
	\]
	where $\frac{\delta C}{\delta m}= \bd m$ is the vorticity 2-form, $\bd m=({\rm curl}\bs{m})\cdot d\bS$ and . 
\end{proposition}
\begin{proof}
	In vector calculus terms, one computes
	\[
	\SCP{\ad^*_u m}{{\frac{\delta C}{\delta m}}}
	=
	\Scp{- \bs{u}\times {\rm curl}\bs{m} + \nabla (\bs{u\cdot m}) }
	{ {\rm curl}\bs{m}}
	= 
	\oint_{\p\mathcal{D}}(\bs{u\cdot m}) \,{\rm curl}\bs{m}\cdot \bs{\widehat{n}} dS
	= 0\,,
	\]
	which is obtained after an integration by parts in the second summand and taking the boundary term to vanish as usual for the helicity.
\end{proof}

\begin{example}[Proto-SALT SFLP equation]
	Let us choose $F_i = \ad^*_{f_i(x)} m$ in equation \eqref{SFLP-general-noise}, where $f_i \in \mathfrak{X}$ are arbitrary functions taking values in the Lie algebra of smooth vector fields $\mathfrak{X}$. Then, the SFLP equation \eqref{SFLP-general-noise} recovers an expression similar to the momentum terms in the SALT fluid motion equation \eqref{eq:SLP eq SALT}; namely,
	\begin{align}
	{\rmd} m + {\ad}^*_{\rmd x_t }m  = 0\,, \quad \text{where } u = \frac{\delta h}{\delta m}
	\quad\hbox{and}\quad
	\rmd x_t = u \,dt + \sum_i f_i \circ dW^i_t
	\,.
	\label{LDP-Ham-SALT-Euler}
	\end{align}
	This is the SALT fluid motion equation \eqref{eq:SLP eq SALT} in the absence of advected quantities.
\end{example}\bigskip


\subsection{Energy preserving SFLT for semidirect-product motion and material entrainment} \label{sec: SDP SFLT-erg}
We next extend the SFLT approach to case of coadjoint motion of fluids carrying advected quantities such as mass and heat which are associated with potential energy.\,\footnote{Advected quantities are also known as order parameters in condensed matter physics.} The motion of advected quantities arises from the semidirect-product action of the diffeomorphisms on the vector spaces containing the advected quantities. The vector spaces containing the advected quantities comprise coset spaces obtained from symmetry-breaking of the full diffeomorphism group to the remaining isotropy subgroup of the initial conditions of the advected quantities. 

Keeping in sight the objective of the paper to derive a class of energy-preserving stochastic models of fluid dynamics, the same energy-preserving noise as the Euler fluids in Example \ref{Prop: Energy-conserved} and its vortex force \ref{Cor: VortexForce} will be chosen in the construction that follows. This extension enables the derivation of stochastic vortex forces which model the uncertainty of unresolved slow-fast interaction effects as energy-preserving stochastic perturbations of fluid models which possess a potential energy. In addition to vortex forces, the introduction of stochasticity in the passive advection relations allows the modelling of quantities that do not passively follow the drift velocity field. Physical applications of this type of stochastic modification are discussed in Remark \ref{remark material entrainment}
For symmetry-reduced semidirect-product motion, considering the energy preserving noise of the form $F_i = \ad^*_u f_i$, the stochastic RLDP principle in \eqref{SRLDP action on G} becomes
\begin{align}
0=\delta S = \delta \int^b_a \ell(u,a)\,dt + \SCP{m}{\rmd g\, g^{-1} - u\,dt} + \SCP{ \rmd b}{a_0g^{-1} - a} - \int_a^b \sum_i\SCP{\ad^*_u f_i}{\delta g\,g^{-1}}\circ dW^i_t,
\label{LDPaction w a}
\end{align}
where the variations $\delta g, \delta u, \delta m, \delta \rmd b, \delta a$ are arbitrary with $\delta g$ vanishing at the endpoints, $t=a,b$. 
Taking the indicated variations in \eqref{LDPaction w a} yields the following result,
\begin{align*}
0&=\int_a^b \SCP{\frac{\delta \ell}{\delta u}}{\delta u}\,dt + \SCP{\frac{\delta \ell}{\delta a}}{\delta a}\,dt + \SCP{\delta m}{\rmd g\,g^{-1} - u\,dt} + \SCP{ \rmd b}{a_0\delta g^{-1}- \delta a} + \SCP{\delta  \rmd b}{a_0 g^{-1}-a}\\
& \quad + \SCP{m}{\delta(\rmd g\, g^{-1}) -  u\,dt} - \sum_i\SCP{\ad^*_u f_i}{\delta g\,g^{-1}}\circ dW_t^i \\
& = \int_a^b \SCP{\frac{\delta \ell}{\delta u}-m}{\delta u}\,dt + \SCP{\frac{\delta \ell}{\delta a}}{\delta a}\,dt + \SCP{\delta m}{\rmd g\,g^{-1} - u\,dt} + \SCP{\rmd  b}{-a\eta- \delta a} + \SCP{\delta  \rmd b}{a_0 g^{-1}-a}\\
& \quad + \SCP{m}{\rmd \eta - \ad_{\rmd g\,g^{-1}}\eta } - \sum_i\SCP{\ad^*_u f_i}{\eta}\circ dW_t^i ,
\end{align*}
where we denote $\eta = \delta g\,g^{-1}$, as before. Vanishing of the coefficients of the variations implies the following relations,
\begin{align}
\begin{split}
&\frac{\delta \ell}{\delta u} = m, \quad \rmd g\,g^{-1} = u\,dt, \quad  \rmd b = \frac{\delta \ell}{\delta a}\,dt \,, \\
&\rmd m = - \ad^*_{\rmd g\,g^{-1}}m +  \rmd b \diamond a - \sum_i\ad^*_u f_i \circ dW^i_t, \quad \rmd a = -\mathsterling_{\rmd g\,g^{-1}}a\,.
\end{split} \label{eq:SFEP without entrainment}
\end{align}
Assembling these relations yields the SFEP equations with advected quantities
\begin{align}
\rmd \frac{\delta \ell}{\delta u} = - \ad^*_{u}\frac{\delta \ell}{\delta u}\,dt - \sum_i \ad^*_u f_i \circ dW^i_t + \frac{\delta \ell}{\delta a} \diamond a \,dt, \quad \rmd a = -\mathsterling_{u} a \,dt\,.
\end{align}
Now, an application of the Legendre transform $h(m,a) = \scp{m}{u} - \ell(u,a)$, followed by calculations similar to those made in the deterministic case arrives at the SFLP equations with advected quantities
\begin{align}
\rmd m +  \ad^*_{\frac{\delta h}{\delta m}}m\,dt + \sum_i \ad^*_{\frac{\delta h}{\delta m}} f_i \circ dW^i_t + \frac{\delta h}{\delta a} \diamond a \,dt = 0, \quad \rmd a = -\mathsterling_{u} a \,dt\,.
\label{eq:SAPS force m only}
\end{align}
These equations can also be written in the form of a Poisson operator, as
\begin{align}
\rmd \begin{bmatrix}m \\ a \end{bmatrix} = -
\begin{bmatrix}
\ad^*_{\fbox{}} \left(m\,dt + \sum_i f_i \circ dW^i_t\right) & {\fbox{}}\diamond a\,dt \\
\mathsterling_{\fbox{}}\,a\,dt & 0 
\end{bmatrix}
\begin{bmatrix}
{\delta h}/{\delta m} \\ {\delta h}/{\delta a} 
\end{bmatrix}\,.
\end{align}
\begin{remark}[Alternative formulation for a reduced Hamilton-d'Alembert principle with advected quantities]
	The SFLP equation with advected quantities \eqref{eq:SAPS force m only} can also  be derived from the following reduced Hamilton-d'Alembert phase space principle with advected quantities, 
	\begin{align*}
	0=\delta S = \delta \int^b_a \SCP{m}{\rmd g\, g^{-1}} + \SCP{\rmd b}{a_0g^{-1} - a} - h(m,a)\,dt - \int^b_a \sum_i \SCP{\ad^*_{\frac{\delta h}{\delta m}}f_i}{\delta g\,g^{-1}}\circ dW^i_t.
	\end{align*}
	The proof of this statement is a direct calculation following the same pattern as for the EP derivation.
\end{remark}
\begin{theorem}[SFLT Kelvin-Noether theorem]\label{thm:SFLT KN}
    The Kelvin-Noether quantity $\SCP{\mathcal{K}(g(t)c_0, a(t))}{m(t)}$ associated with equation \eqref{eq:SAPS force m only} satisfies the following stochastic Kelvin-Noether relation.
    \begin{align}
    \rmd \SCP{\mathcal{K}(g(t)c_0, a(t))}{m(t)} = \SCP{\mathcal{K}(g(t)c_0, a(t))}{- \frac{\delta h}{\delta a}\diamond a\,dt - \sum_i \ad^*_u f^m_i \circ dW^i_t},
    \end{align}
    where the identification $\rmd g\,g^{-1} = \frac{\delta h}{\delta m}\,dt$ is obtained from the stochastic RLDP  constrained variational principle in \eqref{LDPaction w a}.    
\end{theorem}
\begin{remark}\label{thm:SFLT KC}
    Upon assuming that the fluid density $D$ is also advected by the flow, so that $\partial_t D + \mathcal{L}_u D = 0$, the Kelvin circulation theorem may be expressed as
    \begin{align}
    \rmd \oint_{c(u)} \frac{m}{D} = \oint_{c(u)} \frac{1}{D}\left[ - \frac{\delta h}{\delta a}\diamond a\,dt - \sum_i \ad^*_u f^m_i \circ dW^i_t\right], 
    \end{align}
    where the notation of Lagrangian loop $c(u)$ denotes the Lagrangian loop moving with the deterministic fluid velocity $u$, as in the deterministic case. 
\end{remark}
\paragraph{Natural generalisation of SFLT} 
Note that the SFLP equations with advected quantities in \eqref{eq:SAPS force m only} modifies the momentum equation alone, while keeping the advection equation of $a\in V^*$ unchanged. In terms of co-adjoint motion, the natural generalisation of \eqref{SAPS-LP-eq} to semidirect-product Lie group action would be
\begin{align}
\rmd(m,a) = - \ad^*_{\left(\frac{\delta h}{\delta m},\frac{\delta h}{\delta a}\right)}(m,a)\,dt - \sum_i \ad^*_{\left(\frac{\delta h}{\delta m},\frac{\delta h}{\delta a}\right)}(f^m_i,f^a_i)\circ dW_t^i\,,
\label{eq:SFLP semidirect}
\end{align}
where $(m,a), (f^m_i, f^a_i) \in \mathfrak{s}^*$, $\left(\frac{\delta h}{\delta m},\frac{\delta h}{\delta a}\right)\in\mathfrak{s}$ and $\ad^*$ is the coadjoint operator on $\mathfrak{s}^*$. The maps $f^m_i:\mathfrak{X}\times V^* \rightarrow \mathfrak{X}^*$ and $f^a_i:\mathfrak{X}\times V^* \rightarrow V^*$ are arbitrary for all $i$. Using definition of semidirect product coadjoint action discussed in appendix \ref{app: semidirect ad^*}, the individual equations are then
\begin{align}
    \begin{split}
    &\rmd m + \ad^*_{\frac{\delta h}{\delta m}} m\,dt + \frac{\delta h}{\delta a}\diamond a\,dt + \sum_i \ad^*_{\frac{\delta h}{\delta m}} f^m_i\circ dW^i_t + \frac{\delta h}{\delta a}\diamond f^a_i\circ dW^i_t = 0\,, \\
    & \rmd a + \mathcal{L}_{\frac{\delta h}{\delta m}} a\,dt + \sum_i \mathcal{L}_{\frac{\delta h}{\delta m}} f^a_i\circ dW_t^i = 0\,.
    \end{split} 
    \label{ExplicitSDPergEqns}
\end{align}
These individual equations can be written in the form of a Poisson operator, as follows,
\begin{align}
\rmd \begin{bmatrix}m \\ a \end{bmatrix} = -
\begin{bmatrix}
\ad^*_{\fbox{}} \left(m\,dt + \sum_i f^m_i \circ dW^i_t\right) & \fbox{}\diamond \left(a\,dt + \sum_i f^a_i\circ dW_t^i\right) \\
\mathsterling_{\fbox{}}\left(a\,dt + \sum_i f^a_i\circ dW_t^i\right) & 0 
\end{bmatrix}
\begin{bmatrix}
{\delta h}/{\delta m} \\ {\delta h}/{\delta a} 
\end{bmatrix}
\,.\label{eq:SPO new}
\end{align}
Energy is conserved since the Poisson operator is skew-symmetric. Thus, 
\begin{align*}
\rmd h(m,a) &= \SCP{\rmd (m,a)}{\left(\frac{\delta h}{\delta m}, \frac{\delta h}{\delta a}\right)} 
\\&= \SCP{\ad^*_{\left(\frac{\delta h}{\delta m},\frac{\delta h}{\delta a}\right)}(m,a)\,dt - \sum_i \ad^*_{\left(\frac{\delta h}{\delta m},\frac{\delta h}{\delta a}\right)}(f^m_i,f^a_i)\circ dW_t^i}{\left(\frac{\delta h}{\delta m}, \frac{\delta h}{\delta a}\right)} = 0,
\end{align*}
where the last equality uses the anti-symmetry of the $\ad$ operator. The class of SFLP equations in \eqref{eq:SFLP semidirect} can be obtained via a phase-space variational principle, as we discuss next. 

\paragraph{Variational principle for semidirect product SFLP equation}
The SFLP equation \eqref{eq:SFLP semidirect} can be derived from a reduced Hamilton-d'Alembert phase space variational principle in terms of the full semidirect product group $S = G\circledS V$ with the associated semidirect-product Lie algebra $\mathfrak{s}=\mathfrak{X}\circledS V$. This phase-space variational principle reads
\begin{align}
\delta \int_{t_1}^{t_2} \SCP{(m,a)}{(\rmd u, \rmd b)}_\mathfrak{s} - h(m,a) \,dt - \sum_i \int_{t_1}^{t_2}\SCP{\ad^*_{(\frac{\delta h}{\delta m}, \frac{\delta h}{\delta a})}(f^m_i, f^a_i)}{(\eta, w)}_\mathfrak{s} \circ dW^i_t = 0\,, \label{eq: SDP varprinc}
\end{align}
for arbitrary variations of $\delta (m,a)$ of $(m,a) \in \mathfrak{s}^*$ and constrained variations $\delta (\rmd u,\rmd b)$ of $(\rmd u, \rmd b)\in \mathfrak{s}$. The constrained variation of $(\rmd u, \rmd b)$ takes the form 
\begin{align*}
    \delta \rmd u = \rmd \eta + \qv{\rmd u}{\eta}
    \quad\hbox{and}\quad
    \delta \rmd b = \rmd  w - \rmd b\,\eta + w \rmd u\,,
\end{align*}
in which $\eta \in \mathfrak{X}$ and $w\in V$ are arbitrary. 
Here the notation $\SCP{\cdot}{\cdot}_\mathfrak{s}$ is the semidirect product pairing where 
\begin{align*}
    \SCP{(m,a)}{(u,b)}_\mathfrak{s} := \SCP{m}{u}_\mathfrak{X} + \SCP{a}{b}_V.
\end{align*}
In the following we will continue to suppress the pairing subscript when the context is clear. Note that the constrained variations of $\rmd u$ and $\rmd b$ are related to the adjoint action of $\mathfrak{s}$ by
\begin{align*}
    \delta (\rmd u, \rmd b) = \left(\rmd \eta + \qv{\rmd u}{\eta}, \rmd  w - \rmd b\,\eta + w \rmd u\right) =  \rmd (\eta, w) - \ad_{(\rmd u, \rmd b)}(\eta ,w).
\end{align*}
Computing the variations and applying the constrained variations yields
\begin{align*}
    0 = &\int_{t_2}^{t_1} \SCP{\rmd u - \frac{\delta h}{\delta m}\,dt}{\delta m} + \SCP{\rmd b - \frac{\delta h}{\delta a}\,dt}{\delta a} + \SCP{m}{\rmd \eta + \qv{\rmd u}{\eta}} + \SCP{a}{\rmd w - \rmd b\eta + w\,\rmd u} \\
    &\quad - \SCP{\ad^*_{\frac{\delta h}{\delta m}}f^m_i + \frac{\delta h}{\delta a}\diamond f^a_i}{\eta} \circ dW^i_t - \SCP{\mathcal{L}_{\frac{\delta h}{\delta m}}f^a_i}{w} \circ dW^i_t\\
    = &\int_{t_2}^{t_1} \SCP{\rmd u - \frac{\delta h}{\delta m}\,dt}{\delta m} + \SCP{\rmd b - \frac{\delta h}{\delta a}\,dt }{\delta a} + \SCP{-\rmd  m - \ad^*_{\rmd u} m}{\eta} + \SCP{-\rmd a - a\, \rmd u}{w} + \SCP{a\diamond \rmd b}{\eta} \\
    &\quad - \SCP{\ad^*_{\frac{\delta h}{\delta m}}f^m_i + \frac{\delta h}{\delta a}\diamond f^a_i}{\eta} \circ dW^i_t - \SCP{\mathcal{L}_{\frac{\delta h}{\delta m}}f^a_i}{w} \circ dW^i_t\,.
\end{align*}
Consequently, one may collect terms to find the following system of motion and advection equations
\begin{align}
    \begin{split}
    &\rmd m + \ad^*_{\rmd u} m + \rmd  b\,\diamond a +  \ad^*_{\frac{\delta h}{\delta m}} f^m_i\circ dW^i_t + \frac{\delta h}{\delta a}\diamond f^a_i\circ dW^i_t = 0\,, \\
    & \rmd u = \frac{\delta h}{\delta m}\,dt,\quad \rmd b =\frac{\delta h}{\delta a}\,dt, \quad \rmd a + \mathcal{L}_{\rmd u} a + \mathcal{L}_{\frac{\delta h}{\delta m}} f^a_i\circ dW_t^i = 0\,.
    \end{split} 
\end{align}

\begin{remark}[Reduced Hamilton-d'Alembert Pontryagin phase space variation principle]
    The choice of variation principle \eqref{eq: SDP varprinc} is not the Hamiltonian version of the RLDP variation principle used previous sections. It is, however, equivalent to the reduced Hamilton-d'Alembert Pontryagin phase space variation principle 
    \begin{align}
        \delta \int_{t_1}^{t_2} \SCP{(m,a)}{\rmd (g,v)\,(g,v)^{-1}}_\mathfrak{s} - h(m,a) \,dt - \sum_i \int_{t_1}^{t_2}\SCP{\ad^*_{(\frac{\delta h}{\delta m}, \frac{\delta h}{\delta a})}(f^m_i, f^a_i)}{\delta(g,v)\,(g,v)^{-1}}_\mathfrak{s} \circ dW^i_t = 0\,, 
    \end{align}
    where the variations $\delta (g,v) \in S$ and $\delta (m,a)\in \mathfrak{s}$ are arbitrary.
\end{remark}

\begin{proposition}[It\^o form of semidirect product SFLP equation]
    The It\^o form of \eqref{ExplicitSDPergEqns} are
    \begin{align}
    \begin{split}
    &\rmd m + \ad^*_{\frac{\delta h}{\delta m}}\left(m\,dt + \sum_i f^m_i\, dW_t^i \right)+ \frac{\delta h}{\delta a} \diamond \left(a\, dt + \sum_i f^a_i\, dW_t^i\right) + \frac{1}{2}\sum_i\left(\ad^*_{\sigma_i}f^m_i + \theta_i\diamond f^a_i\right)\,dt= 0, \\
    &\rmd a + \mathsterling_{\frac{\delta h}{\delta m}} \left(a\,dt + \sum_i f^a_i\, dW_t^i\right) + \frac{1}{2}\sum_i\mathcal{L}_{\sigma_i}f^a_i\,dt= 0,
    \end{split}
    \end{align}
    where one defines
    \begin{align*}
    \sigma_i := \qp{\frac{\delta^2 h}{\delta m^2}}{-\ad^*_{\frac{\delta h}{\delta m}}f^m_i - \frac{\delta h}{\delta a}\diamond f^a_i}
    \quad \hbox{and}\quad
    \theta_i := \qp{\frac{\delta^2 h}{\delta a^2}}{-\mathcal{L}_{\frac{\delta h}{\delta m}}f^a_i}.
    \end{align*}
    Again the notation $\qp{\cdot}{\cdot}$ denotes contraction, not $L^2$ pairing.    
\end{proposition}
The proof is similar to the case without advected quantities are it is given in appendix \ref{app: Ito form}.
\begin{remark}
The external forces $\ad_{(\frac{\delta h}{\delta m}, \frac{\delta h}{\delta a})}(f^m_i, f^a_i)$ introduced in \eqref{eq: SDP varprinc} are energy preserving only. For a general sets of forces $F^m_i$ and $F^a_i$, the semidirect product SFLP equation will be derived from the variational principle 
\begin{align*}
    \delta \int_{t_1}^{t_2} \SCP{(m,a)}{(\rmd u, \rmd b)} - h(m,a) \,dt - \sum_i \int_{t_1}^{t_2}\SCP{(F^m_i, F^a_i)}{(\eta, w} \circ dW^i_t = 0\,,
\end{align*}
where the variations have the same condition as principle \eqref{eq: SDP varprinc}.
The resulting equations are the following
\begin{align*}
    &\rmd m + \ad^*_{\frac{\delta h}{\delta m}} m\,dt + \frac{\delta h}{\delta a}\diamond a\,dt + \sum_i F^m_i \circ dW^i_t = 0\,, \quad 
    \rmd a + \mathcal{L}_{\frac{\delta h}{\delta m}} a\,dt + \sum_i F^a_i\circ dW_t^i = 0\,.
\end{align*}
\end{remark}

\begin{remark}[Stochastic material entrainment modelled in equation \eqref{ExplicitSDPergEqns}]
	When stochastic processes $f^a_i$ are included in the Lie-derivative action on the fluid variables, $a(t)$, then one can no longer say that $a(t)$ is passively advected by the flow $g \in G$, i.e., $a(t) \neq a_0 g^{-1}(t)$. Consider the case that the variable $a(t) $ represents the mass density of the fluid. Then, the class of stochastic equations in \eqref{ExplicitSDPergEqns} or \eqref{eq:SPO new} could model a fluid containing parcels whose density does not quite passively follow the drift velocity flow. Examples of such deviations from passive transport might include inertial fluid parcels whose density has a certain probability of being heavier or lighter than the ambient (or average, or expected) density. The motion of these inertial parcels would then be uncertain, as modelled by a stochastic variation in their density, relative to parcels undergoing passive advection by the flow. We include this feature of equation \eqref{ExplicitSDPergEqns} because it may provide a useful single-fluid approach to dealing with stochastic entertainment of material particles into fluid flows such as Langmuir circulations. This stochastic model of material entrainment into fluid flows introduces probabilistic aspects into the theory, rather than dealing with the intricacies of multiphase flow models. The applications of this feature may include, for example, ice slurries in the Arctic Ocean, or dust clouds, or debris in tornadoes, or fluid flows with gas bubbles, or well-mixed oil spills, or transport of algae, or plastic detritus in the ocean. In the case where $f^a_i = 0$ for all $i$, then $a(t) $ would be passively advected and would satisfy the standard relation $a(t) = a_0 g^{-1}(t)$. For a recent review of deterministic LES turbulent models of this type of mixed-buoyancy fluid transport, see \cite{CCM2019}. \label{remark material entrainment}
\end{remark}
    
\begin{remark}[SFLT Kelvin-Noether theorem]
    Note that the inclusion of forcing terms $f^a_i$ on advected quantities $a$ implies $a(t) \neq a_0g^{-1}$. This means one cannot take the pull back the Kelvin-Noether quantity $\SCP{\mathcal{K}(g(t)c_0, a(t))}{m(t)}$ by $g(t)$ to the Kelvin-Noether quantity defined by the initial conditions $c_0$ and $a_0$, i.e. $\SCP{\mathcal{K}(g(t)c_0, a(t))}{m(t)}\neq \SCP{\mathcal{K}(c_0, a_0)}{\Ad^*_g m(t)}$. Hence there is no simple modification of the Kelvin-Noether theorem associated with \eqref{ExplicitSDPergEqns}. For fluids, the map $\mathcal{K}$ is the circulation integral around a material loop $c$ which is independent of the advected quantities $a$. Thus a Kelvin circulation theorem exists as formulated below.
\end{remark}

\begin{theorem}[SFLT Kelvin circulation theorem with entrainment]
    Upon assuming that the fluid density $D$ is also advected by the flow, so that $\partial_t D + \mathcal{L}_u D = 0$, the Kelvin circulation theorem associated with \eqref{ExplicitSDPergEqns} may be expressed as
    \begin{align}
    \rmd \oint_{c(u)} \frac{m}{D} = \oint_{c(u)} \frac{1}{D}\left[- \sum_i \ad^*_u f^m_i \circ dW^i_t - \frac{\delta h}{\delta a}\diamond \left(a\,dt + \sum_i f^a_i\circ dW^i_t\right)\right], 
    \label{Kelvin-entrain}
    \end{align}
    where the notation of Lagrangian loop $c(u)$ denotes the Lagrangian loop moving with the deterministic fluid velocity $u$, as in the deterministic case. 
\end{theorem}

\begin{remark}\label{rem: entrain}
	Note that the presence of the stochastic material entrainment terms $f^i_a$ in equation \eqref{Kelvin-entrain} may have a significant effect as a source of circulation of the fluid flow. 
\end{remark}

\begin{remark}[An unified variational approach]\label{remark: Unified approach}
	By introducing stochastic Hamiltonians into the variational principle \eqref{eq: SDP varprinc}, one can formulate a stochastic RLDP principle which encompasses both SALT and SFLT. The augmented stochastic RLDP principle becomes
	\begin{align}
	\begin{split}
            &\delta \int_{t_1}^{t_2} \SCP{(m,a)}{(\rmd u, \rmd b)} - h(m,a) \,dt - \sum_i h_i(m,a)\circ dW^i_t \\
            &\qquad \qquad - \sum_i \int_{t_1}^{t_2}\SCP{(F^m_i, F^a_i)}{(\eta, w)} \circ dW^i_t = 0\,.
	\end{split}
	\label{LDP+SALT action}
	\end{align}
	Here variations of $\delta (m,a)$ of $(m,a) \in \mathfrak{s}^*$ are arbitrary and variations $\delta (g,v)$ of $(g, v)\in S$ are given by
	\begin{align*}
        \delta \rmd u = \rmd \eta + \qv{\rmd u}{\eta}
        \quad\hbox{and}\quad
        \delta \rmd b = \rmd  w - \rmd b\,\eta + w \rmd u\,,
    \end{align*}
	where $(\eta, w) \in \mathfrak{s}$ are arbitrary are vanishes at the boundaries.
	The set of Hamiltonians $h_i$ generates SALT type noise and the set of forces $(F^m_i, F^a_i)$ generates SFLT type of noise. 
\end{remark}

\section{Eulerian Averaged SFLT}\label{sec: EA SFLT}
Ed Lorenz captured the essence of the climate science problem in his celebrated unpublished paper \cite{Lorenz1995} in which he motivated his discussion by invoking the old adage that 
\begin{quote}
``Climate is what you expect. Weather is what you get." 
\end{quote}
Lorenz's lesson was that climate science is fundamentally probabilistic. Much later, this lesson inspired the derivation of the LA SALT fluid model, which also exploited an idea introduced in \cite{DrivasHolm2019} to apply Lagrangian-averaging (LA) in \emph{probability space} to the fluid equations governed by stochastic advection by Lie transport (SALT) which were introduced in \cite{Holm2015}. The general theory of LA SALT and applications to 2D Euler-Boussinesq equations can be found in \cite{DHL2020} and \cite{AdeLHT2020} respectively. 

Here, we apply the probabilistic approach to derive the corresponding Eulerian-averaged SFLT model (EA SFLT) by decomposing the Eulerian solutions of the energy-preserving SFLT model into the sums of their expectations and their fluctuations. As for the energy-preserving SFLT models, the EA SFLT equations admit a Kelvin circulation theorem and preserve the deterministic energy. For systems resulting from quadratic Hamiltonians, this modification of the SFLT model allows the dynamics of the statistical properties of the solutions of EA-SFLT such as the evolution of the expectation of energy to be considered explicitly. As for the LA-SALT models, the EA-SFLT models can be viewed as the interaction of the expected quantities and the fluctuations in a conservative system. In this system, the energy of the expected quantities is dynamically converted into the energy of fluctuations whilst keeping the total energy invariant as shown in Theorem \ref{Thm; Ebal4EA-advec}. The Kelvin circulation theorem for EA-SFLT differs from that of SFLT by the presence of an additional forcing. This feature makes the EA-SFLT approach particularly apt for the examples of quadratic fluid Hamiltonians which are discussed in section \ref{sec: Examples}.
\subsection{EA SFLT}
In Eulerian Averaged SFLT, the expectation of the \emph{Eulerian} quantity $m$ from equation \eqref{SAPS-LP-eq} is the taken over the underlying probability space, and the Eulerian Averaged Stochastic Forced Lie Poisson (EA SFLP) equation is proposed, as
\begin{align}
    \rmd m + \ad^*_u \E{m}\,dt + \sum_i \ad^*_u f^i \circ dW^i_t = 0\,, \quad u = \frac{\delta h}{\delta m}.
    \label{eq:EA SFLP}
\end{align}
This modification still preserves the deterministic energy of the original SFLP equation since
\begin{align}
    \rmd h(m) = \SCP{\rmd m}{u} = -\SCP{\ad^*_u \E{m}\,dt + \sum_i \ad^*_u f^i\circ dW^i_t}{u} = 0, 
\end{align}
by anti-symmetry of of the $\ad$ operation. The evolution of the expectation $\mathbb{E}[m]$ can be determined by considering the It\^o form of the equation \eqref{eq:EA SFLP} 
\begin{align}
    \rmd m + \ad^*_u \E{m}\,dt + \sum_i \ad^*_u f^i\, dW^i_t + \frac{1}{2}\sum_i\ad^*_{\sigma_i} f^i \,dt= 0, 
    \label{eq:EA SFLP Ito}
\end{align}
where $\sigma_i = \qp{\frac{\delta^2 h}{\delta m^2}}{-\ad^*_{\frac{\delta h}{\delta m}}f_i}$. The proof of the It\^o form is similar to the case of the SFLP equation in appendix \ref{app: Ito form}. Taking the expectation of \eqref{eq:EA SFLP Ito} yields the following partial differential equation (PDE) for the expectation of the momentum density, $\E{m}$, 
\begin{align}
    \partial_t \E{m} + \ad^*_{\E{u}} \E{m}+ \frac{1}{2} \sum_i \ad^*_{\E{\sigma_i}}f^i = 0 \,.
    \label{eq:EA SFLP expectation}
\end{align}
The expectation of \eqref{eq:EA SFLP Ito} has produced a deterministic PDE because $\ad^*$ is a linear operation and the expectation of $dW^i_t$ vanishes by It\^o's Lemma. The PDE \eqref{eq:EA SFLP expectation} closes, whenever the Hamiltonian is quadratic in $m$, i.e. $h = \frac12\SCP{m}{\mathbb{I}^{-1}m}$ where $I:\mathfrak{g}\rightarrow \mathfrak{g}^*$ is a constant invertible symmetric operator which commutes with taking the expectation. Furthermore, $\mathbb{I}$ is assumed to be positive definite. In this case, one finds the vector field
\begin{align*}
    \E{\sigma_i} = \qp{-\ad^*_{\E{u}}f^i}{\frac{\delta^2 h}{\delta m^2}} = -\mathbb{I}^{-1} \ad^*_{\E{u}}f^i\,.
\end{align*}
Assuming quadratic Hamiltonian $h$, one can computes also the Hamiltonian function of the expectation of $m$, $\E{m}$ to have 
\begin{align}
    \partial_t h(\E{m}) = \SCP{\E{u}}{\partial_t \E{m}} = -\frac{1}{2}\SCP{\mathbb{I}^{-1} \ad^*_{\E{u}}f^i}{\ad^*_{\E{u}}f^i} < 0\,, \label{eq:EA SFLT h(E(m))}
\end{align}
Define fluctuations $u' = u - \E{u}$, $m' = m - \E{m}$ and $\sigma_i' = \sigma_i - \E{\sigma_i}$ and obtain the fluctuation dynamics of $m$ by subtracting \eqref{eq:EA SFLP expectation} from \eqref{eq:EA SFLP Ito} to obtain
\begin{align}
    \rmd m' + \ad^*_{u'} \E{m}\,dt + \sum_i \ad^*_u f^i\, dW^i_t + \frac{1}{2}\sum_i\ad^*_{\sigma_i'} f^i \,dt= 0 \,.
    \label{eq: fluctdyn}
\end{align}
The time derivative of the Hamiltonian of the fluctuation of $m$, $h(m')$ is found as
\begin{align}
    \begin{split}
        \rmd h(m') &= \SCP{\frac{\delta h}{\delta m'}}{\rmd m'} + \frac{1}{2}\rmd \SCP{m'}{\frac{\delta h}{\delta m'}}_t\\
        &= \SCP{\frac{\delta h}{\delta m'}}{-\ad^*_{u'}\E{m}\,dt - \sum_k \ad^*_u f^k\,dW^k_t - \frac{1}{2}\ad^*_{\sigma_k'}f^k\, dt} + \frac{1}{2}\rmd \SCP{m'}{\frac{\delta h}{\delta m'}}_t\,dt
    \end{split}
\end{align}
This expression simplifies dramatically when the Hamiltonian is quadratic, where 
\begin{align*}
    \frac{\delta h}{\delta m'} = \mathbb{I}^{-1} m' = u', \quad \sigma_i' = -\mathbb{I}^{-1}\ad^*_{u'}f^i.
\end{align*}
A dynamical expression can be calculated as
\begin{align}
    \begin{split}
        \rmd h(m') &= \SCP{u'}{-\ad^*_{u'}\E{m}\,dt - \sum_k \ad^*_u f^k\,dW^k_t - \frac{1}{2}\ad^*_{\sigma_k'}f^k\, dt} + \frac{1}{2}\sum_k\SCP{\ad^*_u f^k}{\mathbb{I}^{-1}\ad^*_u f^k}\,dt\\
        &= \sum_k\SCP{- \ad^*_u f^k\,dW^k_t - \frac{1}{2}\ad^*_{\sigma_k'}f^k\, dt}{u'} + \frac{1}{2}\SCP{\ad^*_u f^k}{\mathbb{I}^{-1}\ad^*_u f^k}\,dt\\
        &= \sum_k\SCP{u'}{-\ad^*_u f^k}\,dW^k_t +\SCP{\sigma_k'}{\frac{1}{2}\ad^*_{u'}f^k}\,dt + \frac{1}{2}\SCP{\ad^*_u f^k}{\mathbb{I}^{-1}\ad^*_u f^k}\,dt\\
        &= \sum_k\SCP{u'}{-\ad^*_u f^k}\,dW^k_t + \SCP{-\mathbb{I}^{-1} \ad^*_uf^k + \mathbb{I}^{-1} \ad^*_{\E{u}} f^k}{\frac{1}{2}\ad^*_{u'}f^k}\,dt + \frac{1}{2}\SCP{\ad^*_u f^k}{\mathbb{I}^{-1}\ad^*_u f^k}\,dt\\
        & = \sum_k\SCP{u'}{-\ad^*_u f^k}\,dW^k_t + \frac{1}{2}\SCP{\mathbb{I}^{-1} \ad^*_u f^k }{\ad^*_{\E{u}}f^k}\,dt - \frac{1}{2}\SCP{\mathbb{I}^{-1} \ad^*_{\E{u}} f^k }{\ad^*_{u}f^k}\,dt \\
        & \qquad \qquad \qquad + \frac{1}{2}\SCP{\mathbb{I}^{-1} \ad^*_{\E{u}} f^k }{\ad^*_{\E{u}}f^k}\,dt\,.
    \end{split}
\end{align}
Taking the expectation yields 
\begin{align}
    \partial_t \mathbb{E}\left[h(m')\right] = \frac{1}{2}\SCP{\mathbb{I}^{-1}\ad^*_{\E{u}}f^k}{\ad^*_{\E{u}}f^k}\,.
\end{align}
The above calculation with \eqref{eq:EA SFLT h(E(m))} have proved the following.
\begin{theorem}[Energy balance for EA SFLT]
The sum of energies $h(\E{m}) + \E{h(m')}$ is preserved by the dynamics of EA SFLT in \eqref{eq:EA SFLP}, 
\begin{align}
    \partial_t h(\E{m}) + \partial_t \E{h(m')} = 0\,.
\end{align}
\end{theorem}
\begin{proof}
The previous calculation demonstrates this theorem by direct calculation. Having done so, an 
alternative proof suggests itself, for quadratic Hamiltonians, one has  $h(m) = h(
\E{m}+m') = const \Longrightarrow \E{h(m)} = h(\E{m}) + \E{h(m')}$.
\end{proof}

\subsection{EA SFLT with advected quantities}
It is straight forward to extend the EA SFLT framework to SFLT systems with advected quantities. Starting with the energy conserving SFLP equation with advected quantities \eqref{ExplicitSDPergEqns}, one can take the average of $m$ and $a$ in the underlying probability space to have the EA SFLP equation with advected quantities
\begin{align}
\begin{split}
    &\rmd m + \ad^*_u \E{m}\,dt + b\diamond \E{a}\,dt + \sum_i \left(\ad^*_u f^i\circ dW^i_t + b\diamond f^a_i\circ dW^i_t\right) = 0\\
    &\rmd a + \mathcal{L}_u \E{a}\,dt + \sum_i \mathcal{L}_u f^a_i \circ dW^i_t = 0\,, \quad (u,b) = \left(\frac{\delta h}{\delta m}, \frac{\delta h}{\delta a}\right)\,,
\end{split} \label{eq:EA SFLP adv}
\end{align}
which can be written more succintly by using the $\ad^*$ action on the semidirect product Lie algebra, more specifically, 
\begin{align}
    \rmd (m,a) + \ad^*_{(u,b)}\mathbb{E}[(m,a)] \,dt + \sum_i \ad^*_{(u,b)}(f^i, f^a_i)\circ dW^i_t = 0\,.
\end{align}
Energy preservation is inherited from the SFLP equation with advected quantities,  
\begin{align}
    \rmd h(m,a) = \SCP{\rmd (m,a)}{\left(\frac{\delta h}{\delta m}, \frac{\delta h}{\delta a}\right)} = -\SCP{ \ad^*_{(u,b)}\mathbb{E}[(m,a)] \,dt + \sum_i \ad^*_{(u,b)}(f^i, f^a_i)\circ dW^i_t}{(u,b)} = 0\,,
\end{align}
where the last equality uses the anti-symmetry of $\ad^*$ of the semidirect product Lie algebra. The It\^o form of \eqref{eq:EA SFLP adv} is
\begin{align}
    \begin{split}
    &\rmd m + \ad^*_u \E{m}\,dt + b\diamond \E{a}\,dt + \sum_i \left(\ad^*_u f^i\, dW^i_t + b\diamond f^a_i\, dW^i_t + \frac{1}{2}\left(\ad^*_{\sigma_i}f^i + \theta_i\diamond f^a_i\right)\,dt \right) = 0, \\
    &\rmd a + \mathcal{L}_u \E{a}\,dt + \sum_i \left(\mathcal{L}_u f^a_i\, dW^i_t + \frac{1}{2}\mathcal{L}_{\sigma_i}f^a_i\,dt\right) = 0,
    \end{split} \label{eq:EA SFLP adv Ito}
\end{align}
where one defines
\begin{align*}
    \sigma_i := \qp{\frac{\delta^2 h}{\delta m^2}}{-\ad^*_{\frac{\delta h}{\delta m}}f^i - \frac{\delta h}{\delta a}\diamond f^a_i}
    \quad \hbox{and}\quad
    \theta_i := \qp{\frac{\delta^2 h}{\delta a}}{-\mathcal{L}_{\frac{\delta h}{\delta m}}f^a_i}.
\end{align*}
The proof is similar to the case of SFLP equation with advected quantities which is included in appendix \ref{app: Ito form}. Taking expectation of \eqref{eq:EA SFLP adv Ito} to have temporal evolution of the expectation of $m$ and $a$. These equations are deterministic because of the linearity of $\ad^*$ and $\diamond$ operators, as well as the $dW^i_t$ terms vanishing by It\^o's Lemma.
\begin{align}
    \begin{split}
        &\partial_t \E{m} + \ad^*_{\E{u}} \E{m} + \E{b}\diamond \E{a} + \frac{1}{2}\sum_i \left(\ad^*_{\E{\sigma_i}}f^m_i + \E{\theta_i}\diamond f^a_i\right) = 0\\
        &\partial_t \E{a} + \mathcal{L}_{\E{u}} \E{a} + \frac{1}{2}\sum_i \mathcal{L}_{\E{\sigma_i}}f^a_i = 0.
    \end{split}\label{eq: EA SFLT E(m) E(a) evo}
\end{align}
Equations \eqref{eq: EA SFLT E(m) E(a) evo} closes when the Hamiltonian is assumed to be quadratic in $m$ and $a$, i.e., $\frac{\delta^2 h}{\delta m^2} = \mathbb{I}^{-1}$ and $\frac{\delta^2 h}{\delta a^2} = \mathbb{J}^{-1}$, where the $\mathbb{I}:\mathfrak{g}\rightarrow \mathfrak{g}^*$ and $\mathbb{J}:V \rightarrow V^*$ are the inertia tensors for the vector fields and advected quantities respectively, which are also assumed to be positive definite. 
\begin{align*}
    \frac{\delta h}{\delta m'} = \mathbb{I}^{-1}m' = u'\,, \quad \frac{\delta h}{\delta a'} = \mathbb{J}^{-1}a' = b'
\end{align*}
In this case one have the vector fields
\begin{align}
\begin{split}
    &\E{\sigma_i} = \qp{-\frac{\delta^2 h}{\delta m^2}}{\ad^*_{\E{u}}f^i + \E{b}\diamond f^a_i} = -\mathbb{I}^{-1}\left(\ad^*_{\E{u}}f^i + \E{b}\diamond f^a_i\right)\\
    &\E{\theta_i} = \qp{-\frac{\delta^2 h}{\delta a^2}}{\mathcal{L}_{\E{u}}f^a_i} = - \mathbb{J}^{-1} \left(\mathcal{L}_{\E{u}}f^a_i\right)\,,
\end{split}
\end{align}
and the equations \eqref{eq: EA SFLT E(m) E(a) evo} closes. Using the evolution of the expectations $\E{m}$ and $\E{a}$, we have the following time derivative of a quadratic Hamiltonian of $\E{m}$ and $\E{a}$
\begin{align}
\begin{split}
    \partial_t h(\E{m}, \E{a}) &= -\sum_i \frac{1}{2}\left(\SCP{\E{u}}{\ad^*_{\E{\sigma_i}}f^m_i + \E{\theta_i}\diamond f^a_i} + \SCP{\E{b}}{\mathcal{L}_{\E{\sigma_i}}f^a_i}\right)\\
    & = -\sum_i \frac{1}{2}\left(\SCP{\mathbb{I}^{-1}\left(\ad^*_{\E{u}}f^i + \E{b}\diamond f^a_i \right)}{\left(\ad^*_{\E{u}}f^i + \E{b}\diamond f^a_i \right)} + \SCP{\mathbb{J}^{-1} \mathcal{L}_{\E{u}}f^a_i}{\mathcal{L}_{\E{u}}f^a_i}\right) < 0\,.
\end{split}\label{eq:EA SFLT H(E(m), E(a)) evo}
\end{align}
The fluctuation $m' = m - \E{m}$ and $a' = a - \E{a}$ are computed in It\^o form using \eqref{eq:EA SFLP adv Ito} and \eqref{eq: EA SFLT E(m) E(a) evo} as 
\begin{align}
\begin{split}
    &\rmd m' + \ad^*_{u'} \E{m}\,dt + b'\diamond \E{a}\,dt + \sum_i \left(\ad^*_u f^i\,dW^i_t + b\diamond f^a_i \, dW^i_t + \frac{1}{2}\left(\ad^*_{\sigma_i'} f^i + \theta_i' \diamond f^a_i\right)\, dt \right) = 0\\
    &\rmd a' + \mathcal{L}_{u'}\E{a}\,dt + \sum_i \left(\mathcal{L}_u f^a_i \,dW^i_t + \frac{1}{2} \mathcal{L}_{\sigma_i'}f^a_i \,dt \right) = 0\,,
\end{split}
\end{align}
where the notations $b'$, $\sigma_i'$ and $\theta_i'$ are defined as $b' = b - \E{b}$, $\sigma_i' = \sigma_i - \E{\sigma_i}$ and $\theta_i' = \theta_i - \E{\theta_i}$.
Then the stochastic time derivative of quadratic Hamiltonians of fluctuations can be computed in similar fashion as in the case without advected quantities. 
\begin{align}
\begin{split}
    \rmd h(m',a') &= \SCP{\rmd m'}{\frac{\delta h}{\delta m'}} + \SCP{\rmd a'}{\frac{\delta h}{\delta a'}} + \frac{1}{2}\rmd \SCP{m'}{\frac{\delta h}{\delta m'}}_t + \frac{1}{2}\rmd\SCP{a'}{\frac{\delta h}{\delta a'}}_t \\
    &= -\sum_i\SCP{\ad^*_u f^i + b\diamond f^a_i}{u'}\,dW^i_t + \SCP{\mathcal{L}_u f^a_i}{b'}dW^i_t \\
    & \qquad - \frac{1}{2}\SCP{\mathbb{I}^{-1}\left(\ad^*_{u'}f^i + b'\diamond f^a_i \right)}{\left(\ad^*_{u'}f^i + b'\diamond f^a_i \right)}\,dt - \frac{1}{2}\SCP{\mathbb{J}^{-1} \mathcal{L}_{u'}f^a_i}{\mathcal{L}_{u'}f^a_i}\,dt\\
    & \qquad + \frac{1}{2}\SCP{\mathbb{I}^{-1}\left(\ad^*_{u}f^i + b\diamond f^a_i \right)}{\left(\ad^*_{u}f^i + b\diamond f^a_i \right)}\,dt + \frac{1}{2}\SCP{\mathbb{J}^{-1} \mathcal{L}_{u}f^a_i}{\mathcal{L}_{u}f^a_i}\,dt
\end{split}
\end{align}
Taking expectation gives 
\begin{align}
    \partial_t \E{h(m',a')} = \sum_i \frac{1}{2}\left(\SCP{\mathbb{I}^{-1}\left(\ad^*_{\E{u}}f^i + \E{b}\diamond f^a_i \right)}{\left(\ad^*_{\E{u}}f^i + \E{b}\diamond f^a_i \right)} + \SCP{\mathbb{J}^{-1} \mathcal{L}_{\E{u}}f^a_i}{\mathcal{L}_{\E{u}}f^a_i}\right)\,.
\end{align}
The above calculations with \eqref{eq:EA SFLT H(E(m), E(a)) evo} proves the following theorem
\begin{theorem}[Energy balance for EA SFLT with advected quantities]\label{Thm; Ebal4EA-advec}
Equations \eqref{eq:EA SFLP adv} lead to the energy balance,
\begin{align}
    \partial_t h(\E{m}, \E{a}) + \partial_t \E{h(m', a')} = 0\,.
    \label{eq: TotalErgCons}
\end{align}

\end{theorem}
\begin{proof}
The previous calculation demonstrates this theorem by direct calculation. Having done so, an 
alternative proof suggests itself for quadratic Hamiltonians, since in this case one has  $h(m,a) = h(
\E{m}+m', \E{a} + a') = const \Longrightarrow \E{h(m, a)} = h(\E{m}, \E{a}) + \E{h(m', a')}$.
\end{proof}
\begin{theorem}[EA SFLT Kelvin circulation theorem]
    Upon assuming that the fluid density $D$ is also advected by the flow, so that $\partial_t D + \mathcal{L}_u D = 0$, the Kelvin circulation theorem may be expressed as
    \begin{align}
    \rmd \oint_{c(u)} \frac{m}{D} = \oint_{c(u)} \mathcal{L}_u \left(\frac{m}{D}\right) \,dt + \frac{1}{D}\left[ \frac{m}{D}\mathcal{L}_u D\, dt - \ad^*_u \E{m}\,dt-b\diamond \E{a}\,dt -  \sum_i \left( \ad^*_u f^m_i \circ dW^i_t + b\diamond f^i_a\circ dW^i_t\right)\right], 
    \label{eq: EA SFLT kelvin}
    \end{align}
    where the notation of Lagrangian loop $c(u)$ denotes the Lagrangian loop moving with the deterministic fluid velocity $u$, as in the deterministic case.
\end{theorem}
\begin{proof}
    When $m\in \mathfrak{X}^*$, one can identify the $\ad^*$ action as the Lie derivative, i.e. $\ad^*_u m = \mathcal{L}_u m$ for all $u \in \mathfrak{u}$. Let $m = D\alpha$ where $\alpha \in \Lambda^1$, the stochastic time derivative of $m$ can be written as
    \begin{align}
        \rmd m &= \alpha \, \rmd D + D \,\rmd \alpha = -\alpha \mathcal{L}_u D\,dt + D \,\rmd \alpha, 
    \end{align}
    thus
    \begin{align}
        \rmd \alpha = \frac{\alpha}{D}\mathcal{L}_u D\,dt -\frac{1}{D} \left(\ad^*_u \E{m}\,dt + b\diamond \E{a}\,dt + \sum_i (\ad^*_u f^i\circ dW^i_t + b\diamond f^a_i\circ dW^i_t)\right)
    \end{align}
    Inserting the expression of $\rmd \alpha$ into the circulation integral to have the result.
    \begin{align}
    \begin{split}
        \rmd \oint_{c(u)}\alpha &= \oint_{c(u)}\left(\rmd + \mathcal{L}_{u\, dt}\right)\alpha \\
        &= \oint_{c(u)} \mathcal{L}_u \alpha \,dt + \frac{\alpha}{D}\mathcal{L}_u D\,dt -\frac{1}{D} \left(\ad^*_u \E{m}\,dt + b\diamond \E{a}\,dt + \sum_i (\ad^*_u f^i\circ dW^i_t + b\diamond f^a_i\circ dW^i_t)\right)
    \end{split}
    \end{align}
\end{proof}
\begin{remark}
This section has responded to Lorentz's lesson in \cite{Lorenz1995} that climate is essentially probabilistic. Accordingly, one might imagine that climate change science would be predicated on the dynamics of variances and higher moments of the fluctuations, which might even apply to the considerations introduced here. Investigations of the potential for applications of the EA SFLT formulation to climate science have been left for future work.
\end{remark}
\section{Examples of SFLT applications}\label{sec: Examples}
\subsection{Heavy Top}\label{subsec: HeavyTop}
The motion of a heavy top under gravity is a good first example application for SFLT, because it is a finite degree-of-freedom subsystem of Euler-Boussinesq fluid motion \cite{Holm1986}.
Following, e.g., \cite{Holm2011}, the configuration space of the heavy top can be taken as the semidirect-product Lie group of Euclidean motions by rotation and translation,  $S = {\rm SO}(3)\circledS \mathbb{R}^3$. The associated Lie algebra and its dual are,  respectively,
\[
\mathfrak{s} = \mathfrak{so}(3)\circledS \mathbb{R}^3 \cong \mathbb{R}^3\circledS \mathbb{R}^3
\quad\hbox{and}\quad 
\mathfrak{s}^* = \mathfrak{so}(3)^*\circledS \mathbb{R}^3 \cong \mathbb{R}^3\circledS \mathbb{R}^3
\,.\]
The natural pairing between $\mathbb{R}^3$ and its dual is the dot-product, for the pairing of Euclidean vectors in  $\mathbb{R}^3$.
The heavy top Hamiltonian $H(\mb{\Pi}, \mb{\Gamma})$ comprises the sum of its rotational kinetic energy and its gravitational potential energy,
\begin{align}
H(\mb{\Pi}) = \frac12 \mb{\Pi} \cdot \mathbb{I}^{-1}\mb{\Pi} + mg\mb{\Gamma}\cdot\mb{\chi}\,.
\label{HT-Ham}
\end{align}
The Hamiltonian for the heavy top is written in the body frame, in terms of the body angular momentum, $\mb{\Pi}$, and the vertical unit vector as seen from the body, $\mb{\Gamma}= \mb{\widehat{z}}O(t)^{-1}$, where $O(t)\in SO(3)$ is the time-dependent rotation from the reference configuration of the top to its current configuration.
Here, $\mathbb{I}$ is the moment of inertia in the body frame and $\mb{\chi}$ is the vector from the point of support to the centre of mass in the body frame. Finally, $m$ is the mass of the body and $g$ is the constant acceleration of gravity. \\
The variational derivatives of the heavy top Hamiltonian are 
\[
{d H}/{d \mb{\Pi}} = \mathbb{I}^{-1}\mb{\Pi} = \mb{\Omega}
\quad\hbox{and}\quad 
{d H}/{d \mb{\Gamma}} = mg\mb{\chi}
\,.\]
Upon using the `hat' map isomorphism  \cite{Holm2011} to identify the Lie derivative and $\ad^*$ operation as the cross product in the $\mathbb{R}^3$ representation, 
the deterministic LP equation is written as
\begin{align*}
\frac{d \mb{\Pi}}{d t} =  \mb{\Pi} \times \mb{\Omega} + \mb{\Gamma}\times mg\mb{\chi}, \quad \frac{d \mb{\Pi}}{d t} =  \mb{\Gamma} \times \mb{\Omega}\,,
\end{align*}
which can also be written in Poisson operator form
\begin{align*}
\frac{d}{dt} \begin{bmatrix} \mb{\Pi} \\ \mb{\Gamma} \end{bmatrix}
= 
\begin{bmatrix}
\mb{\Pi}\times & \mb{\Gamma}\times\\ \mb{\Gamma}\times & 0
\end{bmatrix}
\begin{bmatrix}
\delta H / \delta \mb{\Pi} 
\\
\delta H / \delta \mb{\Gamma}
\end{bmatrix}.
\end{align*}
Defining the force vectors $\mb{f^m_i}, \mb{f^a_i} \in \mathbb{R}^3$, the SFLP equations corresponding to \eqref{ExplicitSDPergEqns} may be written as 
\begin{align}
\begin{split}
\rmd \mb{\Pi} &= \left(\mb{\Pi}\,dt + \sum_i\mb{f^m_i} \circ dW^i_t\right)\times \mb{\Omega} +\left(\mb{\Gamma}\,dt + \sum_i\mb{f^a_i} \circ dW^i_t\right)\times mg\mb{\chi}
\,,\\
\rmd \mb{\Gamma} &= \left(\mb{\Gamma}\,dt + \sum_i\mb{f^a_i} \circ dW^i_t\right)\times \mb{\Omega} \,.   
\end{split}\label{eq:HT SFLP}
\end{align}
When written in Poisson bracket form, these equations become
\begin{align}
\rmd
\begin{bmatrix}
\mb{\Pi} \\ \mb{\Gamma}  
\end{bmatrix}
= 
\begin{bmatrix}
(\mb{\Pi}\,dt + \sum_i\mb{f^m_i} \circ dW^i_t)\times & (\mb{\Gamma}\,dt + \sum_i\mb{f^a_i}\circ dW^i_t)\times\\
(\mb{\Gamma}\,dt + \sum_i\mb{f^a_i}\circ dW^i_t)\times & 0
\end{bmatrix}
\begin{bmatrix}
\delta H / \delta \mb{\Pi} 
\\
\delta H / \delta \mb{\Gamma}
\end{bmatrix}.
\end{align}
One observes that these equations preserve the Hamiltonian in \eqref{HT-Ham}, which follows because of skew symmetry of the matrix Poisson operator. 
\paragraph{Eulerian averaged heavy top}
Consider the EA SFLP equations associated to \eqref{eq:HT SFLP}, which read
\begin{align}
    \begin{split}
    \rmd \mb{\Pi} &= \left(\E{\mb{\Pi}}\,dt + \sum_i\mb{f^m_i} \circ dW^i_t\right)\times \mb{\Omega} + \left(\E{\mb{\Gamma}}\,dt + \sum_i\mb{f^a_i} \circ dW^i_t\right)\times mg\mb{\chi}
\,,\\
\rmd \mb{\Gamma} &= \left(\E{\mb{\Gamma}}\,dt + \sum_i\mb{f^a_i} \circ dW^i_t\right)\times \mb{\Omega} \,.     
    \end{split}
\end{align}
Energy conservation is immediate since the Poisson struture is preserved. In It\^o form, these equation reads 
\begin{align}
    \begin{split}
        \rmd \mb{\Pi} &= \left(\E{\mb{\Pi}}\,dt + \sum_i\mb{f^m_i} \, dW^i_t\right)\times \mb{\Omega} + \left(\E{\mb{\Gamma}}\,dt + \sum_i\mb{f^a_i} \, dW^i_t\right)\times mg\mb{\chi} \\
        & \qquad \qquad + \frac{1}{2}\sum_i \mb{f^m_i}\times \mathbb{I}^{-1}\left(\mb{f^m_i}\times \mb{\Omega} + \mb{f^a_i}\times mg\mb{\chi}\right) \,,\\
        \rmd \mb{\Gamma} &= \left(\E{\mb{\Gamma}}\,dt + \sum_i\mb{f^a_i} \, dW^i_t\right)\times \mb{\Omega} + \frac{1}{2}\sum_i \mb{f^m_i}\times \mathbb{I}^{-1}\left(\mb{f^m_i}\times \mb{\Omega} + \mb{f^a_i}\times mg\mb{\chi}\right)\,.
    \end{split}
\end{align}
Taking expectations to have the time evolution of expectations $\E{\mb{\Pi}}$ and $\E{\mb{\Gamma}}$ as 
\begin{align}
    \begin{split}
        \partial_t \E{\mb{\Pi}} &= \E{\mb{\Pi}} \times \E{\mb{\Omega}} + \E{\mb{\Gamma}} \times mg\mb{\chi} + \frac{1}{2}\sum_i \mb{f^m_i}\times \mathbb{I}^{-1}\left(\mb{f^m_i}\times \E{\mb{\Omega}} + \mb{f^a_i}\times mg\E{\mb{\chi}}\right)\,,\\
        \partial_t \E{\mb{\Gamma}} &= \E{\mb{\Gamma}}\times \E{\mb{\Omega}} + \frac{1}{2}\sum_i \mb{f^m_i}\times \mathbb{I}^{-1}\left(\mb{f^m_i}\times \E{\mb{\Omega}} + \mb{f^a_i}\times mg\E{\mb{\chi}}\right)\,.
    \end{split}
\end{align}
As the Hamiltonian is quadratic, the above equations close and we have 
\begin{align}
    \partial_t H(\E{\mb{\Pi}}, \E{\mb{\Gamma}}) = -\frac{1}{2}\left(\mb{f^m_i}\times \E{\mb{\Omega}} + \mb{f^a_i}\times mg\E{\mb {\chi}}\right)\mathbb{I}^{-1}\left(\mb{f^m_i}\times \E{\mb{\Omega}} + \mb{f^a_i}\times mg\E{\mb{\chi}}\right)\,.
\end{align}
If the inertial tensor $\mathbb{I}$ is positive definite, the energy of the expectations $\E{\mb{\Pi}}$ and $\E{\mb{\Gamma}}$ will decay to zero. Defining the fluctuation of $\mb{\Pi}$ and $\mb{\Gamma}$ as $\mb{\Pi'}:= \mb{\Pi} - \E{\mb{\Pi}}$ and $\mb{\Gamma'}:= \mb{\Gamma} - \E{\mb{\Gamma}}$ respectively, the evolution of the fluctuation can be written in It\^o form as
\begin{align}
    \begin{split}
        \rmd \mb{\Pi'} &= \left(\mb{\Pi'}\,dt + \sum_i\mb{f^m_i} \, dW^i_t\right)\times \mb{\Omega'} + \left(\mb{\Gamma'}\,dt + \sum_i\mb{f^a_i} \, dW^i_t\right)\times mg\mb{\chi'} \\
        & \qquad \qquad + \frac{1}{2}\sum_i \mb{f^m_i}\times \mathbb{I}^{-1}\left(\mb{f^m_i}\times \mb{\Omega'} + \mb{f^a_i}\times mg\mb{\chi'}\right) \,,\\
        \rmd \mb{\Gamma'} &= \left(\E{\mb{\Gamma}}\,dt + \sum_i\mb{f^a_i} \, dW^i_t\right)\times \mb{\Omega'} + \frac{1}{2}\sum_i \mb{f^m_i}\times \mathbb{I}^{-1}\left(\mb{f^m_i}\times \mb{\Omega'} + \mb{f^a_i}\times mg\mb{\chi'}\right)\,.
    \end{split}
\end{align}
The energy of the fluctuations $h(\mb{\Pi'}, \mb{\Gamma'})$ has evoluation equation after using It\^o's Lemma
\begin{align}
    \begin{split}
        \rmd h(\mb{\Pi'}, \mb{\Gamma'}) &= \sum_i \mb{\Omega'}\cdot\left(\mb{f^m_i}\times \mb{\Omega'} + \mb{f^a_i}\times mg\mb{\chi'}\right)\,dW^i_t + mg\mb{\chi'}\cdot \mb{f^a_i} \times \mb{\Omega'}\,dW^i_t\\
        & \qquad \qquad -\frac{1}{2}\left(\mb{f^m_i}\times \mb{\Omega'} + \mb{f^a_i}\times mg\mb{\chi'}\right)\mathbb{I}^{-1}\left(\mb{f^m_i}\times \mb{\Omega'} + \mb{f^a_i}\times mg\mb{\chi'}\right) \\
        & \qquad \qquad +\frac{1}{2}\left(\mb{f^m_i}\times \mb{\Omega} + \mb{f^a_i}\times mg\mb{\chi}\right)\mathbb{I}^{-1}\left(\mb{f^m_i}\times \mb{\Omega} + \mb{f^a_i}\times mg\mb{\chi}\right)
    \end{split}
\end{align}
Taking the expectation yields
\begin{align}
    \partial_t \E{h(\mb{\Pi'} , \mb{\Gamma'})} = \frac{1}{2}\left(\mb{f^m_i}\times \E{\mb{\Omega}} + \mb{f^a_i}\times mg\E{\mb {\chi}}\right)\mathbb{I}^{-1}\left(\mb{f^m_i}\times \E{\mb{\Omega}} + \mb{f^a_i}\times mg\E{\mb{\chi}}\right)\,.
\end{align}
Together with the expression of $\partial_t H(\E{\mb{\Pi}}, \E{\mb{\Gamma}})$, we have exemplified Theorem \ref{Thm; Ebal4EA-advec} for the EA HT dynamics.


\subsection{Energy-preserving stochastic rotating shallow water equations (RSW)}\label{subsec: RSW}

Let $S = \text{Diff}(\mathcal{D})\circledS \text{Den}(\mathcal{D})$, where $\text{Diff}(\mathcal{D})$ denotes the group of diffeomorphisms acting on the planar domain $\mathcal{D}$. Let $\eta\in \mathcal{F}(\mathcal{D})$, and let $\eta\,dV \in \text{Den}(\mathcal{D})$ denote the density on $\mathcal{D}$. The associated Lie algebra and its dual are $\mathfrak{s} = \mathfrak{X}(\mathcal{D})\circledS \text{Den}(\mathcal{D})$ and $\mathfrak{s}^* = (\Lambda^1(\mathcal{D})\otimes \text{Den}(\mathcal{D}))\circledS \text{Den}(\mathcal{D})$.
We denote $\bu\cdot \frac{\partial}{\partial \bx} = u \in \mathfrak{X}(\mathcal{D})$ and $m = \bfm \cdot d\bx\otimes dV \in \Lambda^1(\mathcal{D})\otimes \text{Den}(\mathcal{D})$. In this notation, the LP equations can be written in Cartesian coordinates as the following partial differential equations, 
\begin{align}
    \frac{\partial }{\partial t}\left(\frac{\bfm}{\eta}\right) + (\bu\cdot \GRAD)\frac{\bfm}{\eta} + \frac{m_j}{\eta}\GRAD u^j + \GRAD \frac{\delta H}{\delta \eta} = 0
    \,, \qquad \frac{\partial}{\partial t}\eta + \DIV(\eta\bu) = 0\,,
    \label{vf LP-eq}
\end{align}
and the LP operator can be written as
\begin{align}
    \frac{\partial}{\partial t}
    \begin{bmatrix}
    m_i \\ \eta
    \end{bmatrix}
    = -
    \begin{bmatrix}
    \partial_j m_i + m_j \partial_i & \eta \partial_i \\
    \partial_j \eta & 0  
    \end{bmatrix}
    \begin{bmatrix}
    \delta H / \delta m_j \\
    \delta H / \delta \eta \\
    \end{bmatrix}\,.
    \label{LPBmatrix-op}
\end{align}
For rotating shallow water equations (RSW) in $\mathbb{R}^2$, the Hamiltonian $H$ is given by
\begin{align}
    H = \int_{\mathcal{D}} \frac{1}{2\epsilon\eta}\left|\bfm - \eta \bR \right|^2 + \frac{(\eta - B)^2}{\epsilon \mathcal{F}} \, d^2x\,, \label{RSW Hamiltonian}
\end{align}
in which $\epsilon\ll1$ denotes Rossby number and $\mathcal{F}=O(1)$ denotes the Froude number. The mean depth is $B$ and the surface elevation is $(\eta - B)$.
The reduced Legendre transform yields $\bu = {\delta H}/{\delta \bfm} = (\bfm - \eta\bR)/{\epsilon \eta}$. The variational derivatives of the RSW Hamiltonian are obtained as 
\begin{align*}
    \delta H = \int_{\mathcal{D}} \bu\cdot \delta \bfm + \left(\frac{\eta - B}{\epsilon \mathcal{F}} - \frac{\epsilon}{2}|\bu|^2 - \bu\cdot \bR\right)\cdot \delta \eta \, d^2x\,.
\end{align*}
Substituting into \eqref{vf LP-eq} and using the relation $(\curl \bu) \CROSS \bv + \GRAD(\bu\cdot \bv) = (\bu\cdot\GRAD)\bv + u_i\GRAD v^i$ yields the standard set of RSW equations governing motion and continuity, 
\begin{align}
    \epsilon \frac{\partial \bu}{\partial t} + \left(\curl \bR + \epsilon\, \curl \bu\right)\CROSS \bu + \GRAD\left(\frac{\eta - B}{\epsilon\mathcal{F}} + \frac{\epsilon}{2}|\bu|^2\right) = 0\,, \quad \frac{\partial}{\partial t}\eta + \DIV(\eta\bu) = 0\,.
\end{align}
Consider the Poisson operator of the form \eqref{eq:SPO new} applied to $\mathfrak{s}^*$. It reads
\begin{align}
    \rmd
    \begin{bmatrix}
    m_i \\ \eta
    \end{bmatrix}
    = -
    \begin{bmatrix}
    \partial_j m_i + m_j \partial_i & \eta \partial_i \\
    \partial_j \eta & 0  
    \end{bmatrix}
    \begin{bmatrix}
    \delta H / \delta m_j \\
    \delta H / \delta \eta \\
    \end{bmatrix}\,dt
    - \sum_k
    \begin{bmatrix}
    \partial_j f^k_i + f^k_j \partial_i & g^k \partial_i \\
    \partial_j g^k & 0  
    \end{bmatrix}
    \begin{bmatrix}
    \delta H / \delta m_j \\
    \delta H / \delta \eta \\
    \end{bmatrix}\circ dW_t^k\,, \label{eq:SPO RSW}
\end{align}
where $f^k_i$ are the components of $\bff^k$ such that $f^k = \bff^k \cdot d\bx \otimes dV \in \Lambda^1(\mathcal{D})\otimes \text{Den}(\mathcal{D})$ and $g^k \in \mathcal{F}(\mathcal{D})$ for all $k$. Then the SFLP equations become
\begin{align}
\begin{split}
&\rmd m_i + (\partial_j(m_i u^j) + m_j\partial_i u^j)\,dt + (\partial_j(f^k_i u^j) + \sum_k f^k_j\partial_i u^j)\circ dW_t^k + \eta\partial_i \frac{\delta H}{\delta \eta}\,dt + g^k\partial_i \frac{\delta H}{\delta \eta}\circ dW_t^k = 0, \\
&\rmd \eta + \DIV(\eta\bu)\,dt + \sum_k \DIV(g^k\bu)\circ dW_t^k = 0.
\end{split}
\label{SLP-RSWeqns}
\end{align}
\begin{remark}
	When noise is introduced into the continuity equation for density $\eta$, one can still write the momentum equation in terms of ${\bfm}/{\eta}$. Namely, it reads
	\begin{align}
	\begin{split}
	\rmd \frac{m_i}{\eta}&+ \left(u^j\partial_j\frac{m_i}{\eta}+ \frac{m_j}{\eta}\partial_i u^j\right)\,dt + \partial_i \frac{\delta H}{\delta \eta}\,dt 
	\\& + \sum_k \frac{1}{\eta}\left(- \frac{m_i}{\eta}\partial_j(g^ku^j) + \partial_j(f^k_i u^j) + f^k_j\partial_i u^j + g^k\partial_i \frac{\delta H}{\delta \eta}\right)\circ dW_t^k = 0, 
	\end{split}
	\label{SRSW-eqn}
	\end{align}
	where the inhomogenous term $-\eta^{-2}{m_j}\partial_i(g^k u^j)$ in the second line of \eqref{SRSW-eqn} is due to the modified advection relation for $\eta\, dV$. The SFLP equations for RSW dynamics in \eqref{SLP-RSWeqns} written more succinctly using Lie derivatives and denoting the momentum 1-form as $\alpha = \bfm\cdot d\bx$, then
	\begin{align}
	\begin{split}
	&\rmd \frac{\alpha}{\eta} + \left(\mathcal{L}_u \frac{\alpha}{\eta} + \bd \frac{\delta H}{\delta \eta}\right)\,dt + \sum_k\frac{1}{\eta\, dV} \left(- \frac{\alpha}{\eta}\mathcal{L}_u g^k + \mathcal{L}_u f^k + \frac{\delta H}{\delta \eta}\diamond g^k \right)\circ dW^k_t = 0,\\ 
	&\rmd (\eta\, dV) + \mathcal{L}_u (\eta\, dV)\,dt + \sum_k\mathcal{L}_u (g^k_i\, dV)\circ dW_t^i = 0\,.
	\end{split}\label{eq:SLP 1-form}
	\end{align}
\end{remark}
In vector calculus notation, the stochastic RSW equation obtained by substituting the variational derivatives of $H$ into equation \eqref{eq:SPO RSW} yields
\begin{align}
\begin{split}
&\epsilon \rmd \bu + \left(\curl \bR + \epsilon \,\curl \bu\right)\CROSS \bu\,dt + \GRAD\left(\frac{\eta - B}{\epsilon \mathcal{F}} + \frac{\epsilon}{2}|\bu|^2\right)\,dt \\
& \quad + \sum_k\frac{1}{\eta}\left(-\left(\epsilon \bu +\bR\right)\DIV(g^k\bu) + (\curl \bff^k)\CROSS \bu + \GRAD(\bff^k\cdot \bu) + \bff^k\DIV\bu + g^k\GRAD\pi_{RSW} \right)\circ dW_t^k = 0\,,\\
& \rmd \eta + \DIV(\eta\bu)\,dt + \sum_k\DIV(g^k\bu)\circ dW_t^k = 0,
\end{split} \label{eq:SRSW eq}
\end{align}
where $\pi_{RSW} = \left(\frac{\eta - B}{\epsilon \mathcal{F}} - \frac{\epsilon}{2}|\bu|^2 - \bu\cdot \bR\right)$.

The corresponding Kelvin circulation theorem follows easily from \eqref{eq:SLP 1-form} in geometric form, as
\begin{align}
\rmd \oint_{c(u)} \frac{\alpha}{\eta} = -\oint_{c(u)}\sum_k\frac{1}{\eta\, dV}\left(-\frac{\alpha}{\eta}\mathcal{L}_u g^k + \mathcal{L}_u f^k + g^k d\frac{\delta H}{\delta \eta} \right)\circ dW^k_t\,.
\label{Kel-SRSW}
\end{align}
The stochastic RSW equations in  \eqref{eq:SRSW eq} enable the Kelvin circulation equation \eqref{Kel-SRSW} to be written in vector calculus form as
\begin{align}
&\rmd \oint_{c(u)} (\epsilon\bu + \bR)\cdot d\bx 
\\&\quad= -\oint_{c(u)}\sum_k \frac{1}{\eta}\left(-\left(\epsilon \bu +\bR\right)\DIV(g^k\bu) + (\curl \bff^k)\CROSS \bu + \GRAD(\bff^k\cdot \bu) + \bff^k(\DIV\bu) + g^k\GRAD\pi_{RSW} \right)\cdot d\bx \circ dW_t^k.
\nonumber\end{align}
Here, ones sees the effects of the material entrainment terms proportional to $g^k$ appearing in the generation of Kelvin circulation.
Setting $g^k=0$ simplifies equation \eqref{Kel-SRSW} to 
\begin{align}
\rmd \oint_{c(u)} \frac{\alpha}{\eta} = -\oint_{c(u)}\sum_k\frac{1}{\eta\, dV}
\left( \mathcal{L}_u f^k \right)\circ dW^k_t\,.
\end{align}
The vector form of the motion equation also simplifies for $g^k=0$ to
\begin{align}
\begin{split}
&\epsilon \rmd \bu + \left(\curl \bR + \epsilon\,\curl \bu\right)\CROSS \bu\,dt + \GRAD\left(\frac{\eta - B}{\epsilon \mathcal{F}} + \frac{\epsilon}{2}|\bu|^2\right)\,dt \\
& \qquad = -  \sum_k\frac{1}{\eta}\left( (\curl \bff^k)\CROSS \bu + \GRAD(\bff^k\cdot \bu) + \bff^k\DIV\bu\right)\circ dW_t^k 
\,,
\end{split} \label{eq:SRSW vector eq}
\end{align}
whose right hand side may be regarded as a compressible version of the CL vortex force. The Ito form is 
\begin{align}
\begin{split}
&\epsilon \rmd \bu + \left(\curl \bR + \epsilon\,\curl \bu\right)\CROSS \bu\,dt + \GRAD\left(\frac{\eta - B}{\epsilon \mathcal{F}} + \frac{\epsilon}{2}|\bu|^2\right)\,dt \\
&\qquad +\sum_k\frac{1}{\eta}\left( (\curl \bff^k)\CROSS \bu + \GRAD(\bff^k\cdot \bu) + \bff^k\DIV\bu\right)\, dW_t^k 
 = \frac{1}{2}\sum_k\frac{1}{\eta}\left( (\curl \bff^k)\CROSS \bsigma^k + \GRAD(\bff^k\cdot \bsigma^k) + \bff^k\DIV\bsigma^k\right)\,dt
\,,
\end{split} \label{eq:SRSW vector eq ito}
\end{align}
where $\bsigma^k = \frac{1}{\eta}\left((\curl \bff^k)\CROSS \bu + \GRAD(\bff^k\cdot \bu) + \bff^k\DIV\bu\right)$. By taking the exterior derivative of equation \eqref{eq:SLP 1-form} and noting $\bd \frac{\alpha}{\eta} = \eta q\, dV$ where $q$ is the potential vorticity and $dV$ is the area element, one finds that the vorticity density $\eta q\, dV$ satisfies 
\begin{align}
    \rmd (\eta q \,dV) + \mathcal{L}_u (\eta q \, dV)\,dt + \sum_k \bd \frac{1}{\eta \, dV}\left[-\frac{\alpha}{\eta}\mathcal{L}_u g^k + \mathcal{L}_u f^k + g^k \bd\frac{\delta H}{\delta \eta}\right] = 0\,.\label{eq:RSW-PV-dyn}
\end{align}
\begin{remark}
In coordinates, the last two summands in equation \eqref{eq:RSW-PV-dyn} can be written as divergences, so one finds the pathwise conservation law  $ \rmd \int_{\mathcal{D}}(\eta q \,dV)=0$ for appropriate (homogeneous, or periodic) boundary conditions. 
\end{remark}

\subsection{Euler-Boussinesq (EB) equations}\label{subsec: EB}

Consider the case where $S = \text{Diff}(\mathcal{D})\circledS \mathcal{F}(\mathcal{D})\circledS \text{Den}(\mathcal{D})$, where $D\,dV\in \text{Den}(\mathcal{D})$ and $b \in \mathcal{F}(\mathcal{D})$. The Poisson operator form of the LP equation in this case is 
\begin{align}
\frac{\partial}{\partial t}
\begin{bmatrix}
m_i \\ D \\ b 
\end{bmatrix}
= -
\begin{bmatrix}
\partial_j m_i + m_j \partial_i & D\partial_i & -b_{,i} \\
\partial_jD & 0  & 0 \\
b_{,j} & 0  & 0 
\end{bmatrix}
\begin{bmatrix}
\delta H / \delta m_j \\
\delta H / \delta D \\
\delta H / \delta b
\end{bmatrix}\,.
\end{align}
The stochastic extension of equation \eqref{eq:SPO new} to include buoyancy reads
\begin{align}
\rmd
\begin{bmatrix}
m_i \\ D \\ b 
\end{bmatrix}
= -
\begin{bmatrix}
\partial_j m_i + m_j \partial_i & D\partial_i & -b_{,i} \\
\partial_jD & 0  & 0 \\
b_{,j} & 0  & 0 
\end{bmatrix}
\begin{bmatrix}
\delta H / \delta m_j \\
\delta H / \delta D \\
\delta H / \delta b
\end{bmatrix}\,dt
- \sum_k
\begin{bmatrix}
\partial_j f^k_i + f^k_j \partial_i & g^k\partial_i & -a^k_{,i} \\
\partial_jg^k & 0  & 0 \\
a^k_{,j} & 0  & 0 
\end{bmatrix}
\begin{bmatrix}
\delta H / \delta m_j \\
\delta H / \delta D \\
\delta H / \delta b
\end{bmatrix}\circ dW_t^k\,.
\label{stoch-PB}
\end{align}
Here, $f^k_i$ are the components of $\bff^k$ such that $f^k = \bff^k \cdot d\bx \otimes dV \in \Lambda^1(\mathcal{D})\otimes \text{Den}(\mathcal{D})$, $g^k \in \mathcal{F}(\mathcal{D})$ and $a^k \in \mathcal{F}(\mathcal{D})$ for all $k$. For simplicity, let us consider incompressible flows with $D=1$ and neglect the stochastic part of the advection of $D$. That is, we set $g^k = 0$. This case will yield the standard  incompressibility condition, $\DIV\bu = 0$. 

The EB Hamiltonian is
\begin{align}
H = \int_\mathcal{D} \bigg[ \frac{1}{2D}\big| \bfm - D\bR \big|^2 +  gDbz + p(D-1)  \bigg]d^3x\,,
\label{Ham-det-EB}
\end{align}
where the momentum density $\bfm = D \Big(\bu + \bR(\bx)\Big)$ and the pressure $p$ is a Lagrange multiplier which enforces incompressibility. The variational derivatives of $H$ are given by 
\begin{align}
\delta H = \int_\mathcal{D} \bu\cdot \delta\bfm + \delta D\Big(gbz +p - \frac12|\bu|^2 - \bu\cdot\bR(\bx)\Big) + (gDz)\delta b \,d^3x\,. 
\label{Ham-var-CL}
\end{align}
The deterministic EB equations then follow as
\begin{align} 
\begin{split}    
\partial_t \bu - \bu\times {\rm curl}\big(\bu + \bR(\bx)\big) 
&= -\,\nabla \Big( p + \frac12 |\bu|^2 \Big) - gb\, \mathbf{\hat{z}}  
\, ,   \\
\partial_t D + {\rm div}(D\bu) &= 0\, ,   \quad \hbox{with}\quad D=1
\, ,   \\
\partial_t b + \bu\cdot \nabla b &= 0\,. 
\end{split}   
\label{CL-eqns}   
\end{align}
One can obtain the energy preserving stochastic EB equations by substituting the variational derivatives into stochastic Poisson operator in \eqref{stoch-PB} to find
\begin{align}
\begin{split}
\rmd \bu &- \bu \CROSS \curl(\bu + \bR)\,dt + \GRAD\left(\rmd p + \frac{1}{2}|\bu|^2\,dt\right) + gb\mathbf{\Hat{z}}\,dt + \sum_k \left[ - \bu \CROSS \curl(\bff^k) + \GRAD(\bff^k\cdot\bu) - gz\cdot \GRAD a^k\right]\circ dW_t^k = 0\, ,\\
&\partial_t D + {\rm div}(D\bu) = 0\, ,   \quad \hbox{with}\quad D=1 \, , \\
&\rmd b + \bu\cdot \GRAD b\,dt + \sum_k\bu\cdot\GRAD a^k\circ dW_t^k = 0\,.    
\end{split}\label{eq:SEB eq}
\end{align}
\begin{remark}
	Following \cite{SC2020}, the change of pressure $p \rightarrow \rmd p = p\,dt + \sum_i p_i \circ dW_t^i$ has been made, so that the incompressibility condition, $\DIV\bu = 0$ will be enforced in both the drift and stochastic parts of $\bu$. 
\end{remark}
Similar to the case of $S = \text{Diff}(\mathcal{D})\circledS \text{Den}(\mathcal{D})$ for energy-preserving stochastic RSW equation, a convenient way of considering the Kelvin circulation theorem for the Euler-Boussinesq equations is to write the associated SFLP equations in terms of the momentum 1-form $\alpha = \bfm\cdot d\bx$. These equations then read
\begin{align}
\begin{split}
&\rmd \alpha + \left(\mathcal{L}_u \alpha - \frac{\delta H}{\delta b} \bd b\right)\,dt + \bd\left(\rmd p - \frac{1}{2}|u|^2\,dt - \bu\cdot \bR \,dt\right) + \sum_k\frac{1}{D}\left(\mathcal{L}_u\beta^k - \frac{\delta H}{\delta b}\bd a^k\right)\circ dW^k_t = 0\,,\\
&\partial_t \left(D\, dV\right) + \mathcal{L}_u(D\, dV) = 0\, ,   \quad \hbox{with}\quad D=1 \, , \\
&\rmd b + \bu\cdot \GRAD b\,dt + \sum_k\bu\cdot\GRAD a^k\circ dW_t^k = 0\,,
\end{split}
\label{stoch-EBeqns-erg}
\end{align}
where the 1-form $\beta^k := \bff^k\cdot d\bx$. The Kelvin circulation theorem for these equations is immediate, as
\begin{align}
\rmd\oint_{c(u)} \alpha = \oint_{c(u)} \frac{\delta H}{\delta b}\bd b\,dt - \sum_k\oint_{c(u)} \left(\mathcal{L}_u\beta^k - \frac{\delta H}{\delta b}\bd a^k\right) \circ dW_t^k\,.
\end{align}
In vector calculus notation, this Kelvin circulation theorem is written as
\begin{align}
\rmd\oint_{c(u)}(\bu + \bR)\cdot d\bx = \oint_{c(u)} gz\GRAD b\cdot d\bx \,dt - \sum_k\oint_{c(u)}\left((\curl \bff^k)\CROSS \bu - gz\GRAD a^k \right)\cdot d\bx \circ dW^k_t
\end{align}
\begin{remark}\label{stochCL-remark}
	In the case where $\bff^0 = -\bu^S(\bx)$, $dW_t^0 = dt$, and $a^k=0$ for $k=0,\ldots$, one finds that the energy-preserving stochastic EB equations in \eqref{eq:SEB eq} produce a stochastic contribution to the vortex force in the Craik-Leibovich equations \cite{CL1976} whose formulation with Hamilton's principle is discussed in \cite{HolmCL1996}. Namely, they reduce as follows, 
	\begin{align}
	\begin{split}    
	\rmd \bu - \bu\times {\rm curl}\big(\bu - \bu^S(\bx) + \bR(\bx)\big) dt
	&= -\,\nabla \Big( p + \frac12 |\bu|^2 + \bu\cdot\bu^S \Big)dt  - gb\, \mathbf{\hat{z}} \, dt
	\\&\qquad + \sum_{k>0} \Big( \bu\times {\rm curl}\,\bff^k 
	- \,\nabla \big(\bu\cdot \bff^k\big)\Big) \circ dW^k_t
	\, ,   \\
	\partial_t D + {\rm div}(D\bu) &= 0\, ,   \quad \hbox{with}\quad D=1
	\, ,   \\
	\partial_t b + \bu\cdot \nabla b &= 0\,, 
	\end{split}   
	\label{SCL-eqns}
	\end{align}
	where one interprets the semimartingale $\rmd \bu^S = \bu^S(\bx)\,dt + \sum_k \bff^k(\bx)  \circ dW^k_t $ as a stochastic augmentation of the usual steady prescribed Stokes drift velocity. 
\end{remark}

\paragraph{Physical interpretation of the functions $\bff^k$ and $a^k$ in terms of stochastic PV fluxes}
The energy-preserving stochastic EB equations in \eqref{eq:SEB eq} imply the following equation for potential vorticity density, $q \,dV$, defined by $q\, dV := \bd\alpha \wedge \bd b = (\curl \bfm)\cdot \GRAD b \,dV = (\bs{\omega} + 2 \bs{\Omega})\cdot \GRAD b \,dV$, where $\bfm=\bu+\bR$ for $D=1$. Namely,
\begin{align}
\begin{split}
\rmd (q \,dV)+ \mathcal{L}_u (q\,dV) \,dt &= \bd\left(\rmd\alpha + \mathcal{L}_u\alpha \,dt\right)\wedge \bd b + \bd\alpha \wedge \bd\left(\rmd b + \mathcal{L}_u b \,dt\right)\\
& =  - \sum_k\bd\bigg[\big(\mathcal{L}_u \beta^k) - g z \bd a^k\big) \wedge \bd b -  \big(\mathcal{L}_u a^k\big)\bd \alpha \bigg]\circ dW_t^k
\,,
\end{split}
\label{eqn: PV-StochEB-geom}
\end{align}
where we recall that $\beta^k := \bff^k\cdot d\bx$ is a 1-form.
In vector calculus notation, after using the incompressibility condition $\DIV \bu = 0$, the potential vorticity equation \eqref{eqn: PV-StochEB-geom} can be written in terms of $q = (\curl \,\bfm)\cdot \GRAD b$ as, cf. \cite{{BodnerFK2020}}
\begin{align}
\begin{split}
\rmd q + \bu\cdot\GRAD q\,dt &= \sum_k {\rm div}  \Big[\big(\bu \CROSS \curl\, \bff^k + gz\nabla a^k \big)\times \GRAD b 
+ (\bs{\omega} + 2 \bs{\Omega}) \cdot \GRAD(\bu\cdot \GRAD a^k) 
\Big]\circ dW^k_t
\\ &= \sum_k {\rm div} \Big[\mathbf{F}^k \times  \nabla b 
+ (\bs{\omega} + 2 \bs{\Omega}) \cdot \big(\nabla \mathfrak{D}^k\big)\Big]\circ dW^k_t
=: -\,\sum_k {\rm div} \mathbf{J}^k \circ dW^k_t  
\,,\end{split}
\label{eqn: PV-StochEB-calc}
\end{align}
where $\curl \,\bfm = \bs{\omega} + 2 \bs{\Omega} $ is the total vorticity, and the quantities $\mathbf{F}^k$ and $\mathfrak{D}^k$ are defined as 
\begin{align}
\mathbf{F}^k := \bu \CROSS \curl \,\bff^k + gz\nabla a^k
\quad\hbox{and}\quad
\mathfrak{D}^k := \bu\cdot \GRAD a^k
\,.
\label{eqn: PV-StochEB-FD}
\end{align}
The summands in $\mathbf{J}^k = - \,\mathbf{F}^k \times \nabla b - (\bs{\omega} + 2 \bs{\Omega}) \cdot \big(\nabla \mathfrak{D}^k) $  are called the ``$J$-fluxes of PV'' and are identified with ``frictional'' and ``diabatic'' effects, respectively, in \cite{HM1987, Marshall&Nurser1992}.  See also \cite{BodnerFK2020} for LES turbulence interpretations of these fluxes. 

In summary, while the energy-preserving stochastically-augmented CL vortex force and entrainment effects in equations \eqref{eq:SEB eq} or \eqref{SCL-eqns} can locally create stochastic Langmuir circulations, the total volume-integrated potential vorticity $Q = \int_{\mathcal{D}}q \,dV$ will be preserved for appropriate boundary conditions.
See \cite{SullivMcW-LangTurb2019} for more information about Langmuir circulations and their importance in the mixing processes in the upper ocean boundary layer. The sub-mesoscale excitations created by the $J$-fluxes of PV are a subject of intense present research aimed at understanding the effects of turbulence on oceanic frontogenesis, as well as wave forcing which transports materials such as sediment, gases, algae (carbon), oil spills and plastic detritus, \cite{McWilliams2003,McWilliams2016,McWilliams2017,McWilliams2018,LiFK2017,Bodner-etal, Deng-etal2019, TM-etal2020}. 

\begin{remark}[Eulerian averaged Euler Bousinessq equation]
For completeness, we mention that the corresponding Eulerian averaged equation to the stochastic EB equation \eqref{eq:SEB eq}. Namely, the EA SLFT EB equation is given by
\begin{align}
\begin{split}
    \rmd \bu &- \bu \CROSS \curl\,\E{\bu + \bR}\,dt + \GRAD\left(\rmd p - \frac{1}{2}|\bu|^2\,dt - \bu\cdot \bR\,dt + \E{\bu+ \bR}\cdot\bu\,dt \right) - gz\GRAD \E{b}\,dt \\
    &\qquad + \sum_k \left[ - \bu \CROSS \curl(\bff^k) + \GRAD(\bff^k\cdot\bu) - gz\cdot \GRAD a^k\right]\circ dW_t^k = 0\, ,\\
    &\partial_t D + {\rm div}(D\bu) = 0\, ,   \quad \hbox{with}\quad D=1 \, , \\
    &\rmd b + \bu\cdot \GRAD \E{b}\,dt + \sum_k\bu\cdot\GRAD a^k\circ dW_t^k = 0\,.    
\end{split}\label{eq:EA SFLT EB}
\end{align}
Time evolution of expectation of $\bu$ and $b$ can be found by passing to the Hamiltonian side and we have 
\begin{align}
    \partial_t H(\E{\bfm}, \E{b}) = -\frac{1}{2}\sum_i \int_\mathcal{D} \left|-\E{\bu}\CROSS \curl\bff^i + \GRAD(\E{\bu}\cdot \bff^i) - \E{gz}\cdot \GRAD a^i\right|^2 \,d^3x
\end{align}
and 
\begin{align}
    \partial_t \E{H(\bfm', b')} = \frac{1}{2}\sum_i \int_\mathcal{D} \left|-\E{\bu}\CROSS \curl\bff^i + \GRAD(\E{\bu}\cdot \bff^i) - \E{gz}\cdot \GRAD a^i\right|^2 \,d^3x
\end{align}
which agrees with Theorem \ref{Thm; Ebal4EA-advec}.
\end{remark}

\section{Conclusion and outlook}\label{sec: conclude}
The motivation of this paper has been to determine what type of stochastic perturbations can be added to fluid dynamics that will preserve the fundamental properties of energy conservation, Kelvin circulation theorem and conserved quantities arising from the Lagrangian particle relabelling symmetry. The geometric framework employed in this paper introduces stochastic forcing by Lie transport (SFLT) noise as a series of perturbations which automatically produce a Kelvin circulation theorem and can be chosen to satisfy either energy, or Casimir preservation. In this paper, we have mainly focused on preserving energy conservation. These stochastic external forces can be seen as the slow $+$ fast decomposition of external forces corresponding to the slow $+$ fast decomposition of fluid flow. For Euler fluid equations, the stochastic CL vortex force will always be energy conserving and its physical interpretation as a wave-averaged forces fits well into the external forces considered in the reduced Lagrange-d'Alembert-Pontryagin (RLDP) principle. 
In comparison with the location uncertainty (LU) approach by M\'emin \cite{Memin2014}, the present paper gives an \emph{alternative} set of fluid equations which are energy preserving. The relation to LU has been left for future work. Numerical simulations of the stochastically forced Lie-Poisson (SFLP) equations \eqref{ExplicitSDPergEqns} will be needed to classify solution behaviours of these new stochastic extension of classical fluid equations. As in the applications of the stochastic advection by Lie transport (SALT) and LU approaches, computational simulations of Langmuir fluid circulations and their material entrainment using  equations \eqref{eq:SEB eq} and \eqref{SCL-eqns} will require the calibration of the functions $(f^m_i,f^a_i)$, perhaps via data analysis methods similar to those used in the approach detailed in \cite{CCHOS18a, CCHOS18}. Computational simulations of the equations resulting from the SFLT and SFLP modelling approaches introduced here, as well as simulations of the EA SFLT equations in section \ref{sec: EA SFLT} have all been left for future work. 




\section*{Acknowledgements}
We are grateful to our friends and colleagues who have generously offered their time, thoughts and encouragement in the course of this work during the time of COVID-19. 
Thanks to E. Luesink, S. Takao, W. Pan, D. Crisan, O. Street, F. Gay-Balmaz, E. M\'emin, B. Chapron,  C. Franzke, J. C. McWilliams, B. Fox-Kemper, A. J. Roberts and W. Bauer for their thoughtful comments and discussions. We are also grateful to the anonymous referee for their constructive comments.
DH is also grateful for partial support from ERC Synergy Grant 856408 - STUOD (Stochastic Transport in Upper Ocean Dynamics). RH is supported by an EPSRC scholarship [grant number EP/R513052/1].

\section*{Data availability} 
No data was created or used in writing this paper.

\appendix


\section{Coadjoint operator of semidirect-product Lie-Poisson brackets}\label{app: semidirect ad^*}
Following \cite{HMR1998} and \cite{CHR2018}, consider a Lie Group $G$ which acts from the left by linear maps on a vector space $V$ which induces a left action of $G$ on $V^*$. In the right representation of $G$ on the vector space $V$, the semidirect product group $S = G\circledS V$ has group multiplication
\begin{align}
(g_1, v_1)(g_2,v_2) = (g_1g_2, v_2 + v_1g_2),
\end{align}
where the action of $G$ on $V$ is denoted by concatenation $vg$. The identity element in $S$ is $(e,0)$ where $e$ is the identity in $G$. The inverse of an element in $S$ is given by
\begin{align*}
(g,v)^{-1}  = (g^{-1}, -vg^{-1})
\end{align*}
The Lie algebra bracket on the semidirect product Lie algebra $\mathfrak{s} = \mathfrak{g}\circledS V$ is given by 
\begin{align}
[(\xi_1, v_1),(\xi_2, v_2)] = ([\xi_1, \xi_2], v_2\xi_1 - v_1\xi_2) \label{SDP-ad op}
\end{align}
where the induced action of $\mathfrak{g}$ on $V$ is denoted by concatenation $v\xi$. The operation ${\rm AD}: S\times S\to S$ is defined by
\begin{align}
{\rm AD}_{(g_1,v_1)}(g_2, v_2) = (g_1,v_1)(g_2,v_2)(g_1,v_1)^{-1} = (g_1g_2g_1^{-1}, -v_1g_1^{-1} + v_2g_1^{-1} + v_1g_2g_1^{-1})\,.
\end{align}
Taking the time derivatives of $g_2$ and $v_2$, then evaluating them at the identity $t=0$ yields the Adjoint operation ${\rm Ad}: S \times \mathfrak{s} \to \mathfrak{s}$ which is defined by
\begin{align}
\Ad_{(g,v)}(\xi, a) = \frac{d}{dt}\biggr|_{t=0}{\rm AD}_{(g,v)}(\Tilde{g}(t), \tilde{v}(t)) = (g\xi g^{-1}, (a+v\xi)g^{-1}),
\end{align}
where $\frac{d}{dt}\big|_{t=0}\Tilde{g}(t) = \xi$ and $\frac{d}{dt}\big|_{t=0}\Tilde{v}(t) = a$.
The coAdjoint operation $\Ad^*$ is the formal adjoint of $\Ad$ with respect to the pairings $\big<\cdot,\cdot\big>_\mathfrak{g}$ and $\big<\cdot,\cdot\big>_V$ which can be computed as
\begin{align}
\Ad^*_{(g,v)}(\mu, a) =  (g\mu + (vg^{-1})\diamond(ag^{-1}), ag^{-1}),
\end{align}
where the diamond operator $\diamond$ is defined as 
\begin{align*}
\Scp{\beta \diamond \alpha }{\xi}_\mathfrak{g}:= \Scp{b}{- a \xi }_V.
\end{align*}
The notation $ag^{-1}$ denotes the inverse of the dual isomorphism defined by $g \in G$ (so that $g \rightarrow ag^{-1}$ is a right action). Note that the adjoint and coadjoint actions are left actions. In this case, the $\mathfrak{g}$-actions on $\mathfrak{g}^*$ and $V^*$ are defined as before to be minus the dual map given by the $\mathfrak{g}$-actions on $\mathfrak{g}$ and $V$ and are denoted, respectively, by $\xi\mu$ (left action) and $a\xi$ (right action).
Taking time derivative of $(g,v)$ in the definition of $\Ad_{(g,v)}(\xi,u)$ and evaluating at the identity gives the adjoint operator $\ad$ which coincides with the Lie algebra bracket. One computes the formal adjoint of $\ad$ with respect to the pairing
\begin{align*}
\SCP{\ad_{(\xi, b)} (\tilde{\xi}, \tilde{b})}{(\mu, a)} = \SCP{\left(\ad_\xi \tilde{\xi}, \tilde{b}\xi - b\tilde{\xi}\right)}{(\mu, a)} = \SCP{(\ad^*_\xi \mu + b\diamond a, -a\xi)}{\left(\tilde{\xi}, \tilde{b}\right)} = \SCP{(\ad^*_{(\xi,b)} (\mu,a))}{\left(\tilde{\xi}, \tilde{b}\right)},
\end{align*}
where in the first equality we have used the left Lie algebra action in \eqref{SDP-ad op} to obtain 
\begin{align}
\ad^* _{(\xi,b)}(\mu, a) = (\ad^*_\xi \mu + b\diamond a, -a\xi).
\label{eq: left ad-star}
\end{align}
When $\mathfrak{g} = \mathfrak{X}$, it is formally the right Lie algebra of ${\rm Diff}(\mathcal{D})$, that is, its standard left Lie algebra bracket is minus the usual Lie bracket for vector fields. To distinguish between these brackets, we denote with $[u,v]$ the standard Jacobi-Lie bracket of the vector fields where $u,v\in \mathfrak{X}$ such that $\ad_u v = -[u,v]$. Here $\ad$ operation is the adjoint action of the left Lie algebra to itself. Then the adjoint action of $\mathfrak{s} = \mathfrak{X}\circledS V$ on itself is then
\begin{align}
    \ad_{(\xi_1, v_1)}(\xi_2, v_2) = (-[\xi_1, \xi_2], v_2\xi_1 - v_1\xi_2), \label{eq:spd ad action vf}
\end{align}
and the coadjoint action is taken to be the formal dual of \eqref{eq:spd ad action vf}.
Identifying $-a\xi = \mathcal{L}_\xi a = -\mathcal{L}^T_\xi a$ yields the expression for the Lie-Poisson bracket for the left action of a semidirect-product Lie algebra by simply computing the dual of the semidirect product Lie algebra $\mathfrak{s} = \mathfrak{g}\circledS V$ action given in equation \eqref{SDP-ad op}. Namely, one computes

\begin{align}
\begin{split}
\{f,h \}(\mu, a) &= \SCP{(\mu,a)}{\ad_{(\frac{\delta h}{\delta \mu}, \frac{\delta h}{\delta a})}\left(\frac{\delta f}{\delta \mu},\frac{\delta f}{\delta a}\right)}\\
&=\SCP { \bigg( {\rm ad}^*_{\frac{\delta h}{\delta \mu} } \mu - a\diamond \frac{\delta h}{\delta a} 
	\,,\,-\,\mathcal{L}_{\frac{\delta h}{\delta \mu} }^T a \bigg) } 
{ \bigg( \frac{\delta f}{\delta \mu}\,,\, \frac{\delta f}{\delta a}  \bigg) } \\
&= \SCP{\ad^*_{(\frac{\delta h}{\delta \mu},\frac{\delta h}{\delta a})}(\mu,a)}{\left(\frac{\delta f}{\delta \mu}, \frac{\delta f}{\delta a}\right)}
\,,\end{split}
\end{align}
which of course agrees with \eqref{eq: left ad-star}.

\section{Stochastic advection by Lie transport (SALT)}\label{app: SALT}

\subsection{Deterministic semidirect-product coadjoint motion for ideal fluids}
\begin{proposition}[Euler-Poincar\'e theorem \cite{HMR1998}]
	The Euler-Poincar\'e (EP) equations of a reduced Lagrangian $\ell(u, a)$, $\ell: \mathfrak{X}\times V^* \rightarrow \mathbb{R}$ defined over the space of smooth vector fields with elements $u\in \mathfrak{X}$ acting by Lie derivative on elements $a\in V^*$ of a vector space $V^*$are written as \cite{HMR1998}
	\begin{align}
	\frac{d}{dt}\frac{\delta \ell}{\delta u} + \ad^*_u \frac{\delta \ell}{\delta u} - \frac{\delta \ell}{\delta a}\diamond a = 0, \quad \frac{d}{dt}a + \mathsterling_u a=0.
	\label{EPeqn}
	\end{align}
\end{proposition}\noindent
In the EP equation in \eqref{EPeqn}, the variational derivative is defined as usual by,
\begin{align}
\delta \ell (u,a) = \frac{d}{d\ep}\bigg|_{\ep=0}\ell (u_\ep,a_\ep)
= \left\langle\,\frac{\delta \ell}{\delta u}\,,\,\delta u\,\right\rangle_{\mathfrak{X}}
+ \left\langle\,\frac{\delta \ell}{\delta a}\,,\,\delta a\,\right\rangle_V\,.
\label{VarDeriv-def}
\end{align}
The quantities $a\in V^*$ and $\delta \ell / \delta a \in V$ in \eqref{VarDeriv-def} are dual to each other under the $L^2$ pairing $\langle\,\cdot\,,\,\cdot\,\rangle_V:V\times V^*\to\mathbb{R}$. Likewise, the velocity vector field $u\in \mathfrak{X}$ and momentum 1-form density $m := \delta \ell / \delta u \in \mathfrak{X}^*$ are dual to each other under  the $L^2$ pairing $\langle\,\cdot\,,\,\cdot\,\rangle_\mathfrak{X}:\mathfrak{X}\times \mathfrak{X}^*\to\mathbb{R}$. In terms of these pairings and the Lie derivative operator $\mathcal{L}_u$ with respect to the vector field $u\in \mathfrak{X}$, the coadjoint operator  $\ad_u^*$ and the diamond operator $(\diamond)$ in \eqref{EPeqn} are defined by 
\begin{align}
\Scp{\ad^*_u \frac{\delta \ell}{\delta u} }{v}_\mathfrak{X} 
:= \Scp{\frac{\delta \ell}{\delta u}}{- \mathcal{L}_u v }_\mathfrak{X}
=: \Scp{\frac{\delta \ell}{\delta u}}{\ad_u v }_\mathfrak{X},
\label{ad-star-def}
\end{align}
where $v\in \mathfrak{X}$, and $\ad:\mathfrak{X}\times\mathfrak{X}\to \mathfrak{X}$ is defined as $\ad_u v := - [u,v] := - (u^jv_{,j}^i-v^ju_{,j}^i)\p_i$, and
\begin{align}
\Scp{b \diamond a}{v}_\mathfrak{X} := \Scp{b}{- \mathcal{L}_v a }_V.
\label{diamond-def}
\end{align}
As we shall see later, the coadjoint operator  $\ad_u^*$ and the diamond operator $(\diamond)$ enter together in \eqref{EPeqn} as a form of coadjoint motion for semidirect product action of the Lie algebra 

\begin{proposition}[The EP equation \eqref{EPeqn} also follows from the HPVP]
The EP equation in \eqref{EPeqn} can be derived by direct calculation from the following Hamilton-Pontryagin variational principle (HPVP)
\begin{align}
0 = \delta S = \delta \int^b_a \Big(\ell(u,a) + \SCP{m}{\dot{g}g^{-1} - u} + \SCP{b}{a_0g^{-1} - a}\Big)\, dt,
\label{HPVP-def}
\end{align}
where $m\in \mathfrak{X}^*$, $b\in V$, $g\in G$ and the variations are taken to be arbitrary. 
\end{proposition}

\subsection{Stochastic semidirect-product coadjoint fluid motion with SALT noise}
To add noise in the SALT form to the deterministic HPVP in \eqref{HPVP-def}, we introduce the following stochastic variational principle, cf. \cite{Holm2015},\footnote{Note: we can choose separate (uncorrelated) Brownian motions in the $m$ and $a$ equations in \eqref{StochHPVP-def} by choosing $h_i(m,a) = h^m_i(m) + h^a_i(a)$ for the Stratonovich noise, $\circ dW^i_t$. The choice of Stratonovich noise enables the standard calculus chain rule and product rule to be used for the operations of differentiation and integration by parts, in which variational principles are defined.}
\begin{align}
0=\delta S = \delta \int^b_a \Big(\ell(u,a)\,dt + \SCP{m}{\rmd g\, g^{-1} - u\,dt} + \SCP{ \rmd b}{a_0g^{-1} - a} - \sum_i h_i(m,a)\circ dW^i_t\Big)\,.
\label{StochHPVP-def}
\end{align}
For brevity of notation, we will suppress the subscript labels in the pairings whenever the meaning is clear.  The symbol $\rmd$ in \eqref{StochHPVP-def} abbreviates stochastic time integrations. The action integral $S$ in \eqref{StochHPVP-def} is defined in the framework of variational principles with semimartingale constraints which was established in \cite{SC2020}. As we shall see below, the semimartingale nature of a Lagrange multiplier which imposes one of these semimartingale constraints emerges in the context of the full system of equations, which is obtained after the variations have been taken.   
The Hamiltonian functions $h_i(m,a)$ on $\mathfrak{X}^*\times V^*$  in \eqref{StochHPVP-def} will be prescribed here without discussing how they would be obtained in practice, e.g., via data assimilation. The data assimilation procedure for SALT is discussed, e.g., in \cite{CCHOS18a,CCHOS18}. 

\paragraph{Euler-Poincar\'e (EP) Lagrangian formulation.}
Taking arbitrary variations in the stochastic HPVP in equation \eqref{StochHPVP-def} yields the following determining  relations among the variables, 
\begin{align*}
0&=\int_a^b \SCP{\frac{\delta \ell}{\delta u}}{\delta u}\,dt + \SCP{\frac{\delta \ell}{\delta a}}{\delta a}\,dt + \SCP{\delta m}{\rmd g\,g^{-1} - u\,dt} + \SCP{ \rmd b}{a_0\delta g^{-1}- \delta a} + \SCP{\delta  \rmd b}{a_0 g^{-1}-a}\\
& \quad- \SCP{\frac{\delta h_i}{\delta m}}{\delta m}\circ dW_t^i - \SCP{\frac{\delta h_i}{\delta a}}{\delta a}\circ dW_t^i + \SCP{m}{\delta(\rmd g\, g^{-1}) - \delta u\,dt}\\
& = \int_a^b \SCP{\frac{\delta \ell}{\delta u}}{\delta u}\,dt + \SCP{\frac{\delta \ell}{\delta a}}{\delta a}\,dt + \SCP{\delta m}{\rmd g\,g^{-1} - u\,dt} + \SCP{ \rmd b}{-a\eta- \delta a} + \SCP{\delta \rmd b}{a_0 g^{-1}-a}\\
& \quad- \SCP{\frac{\delta h_i}{\delta m}}{\delta m}\circ dW_t^i - \SCP{\frac{\delta h_i}{\delta a}}{\delta a}\circ dW_t^i + \SCP{m}{\rmd \eta - \ad_{\rmd g\,g^{-1}}\eta - \delta u\,dt}\,.
\end{align*}
Here, $\eta = \delta g\,g^{-1}$ and natural boundary terms have been assumed. Collecting terms among the variational relations gives the following set of equations, which turn out to involve four semimartingales,
\begin{align}
\begin{split}
&\frac{\delta \ell}{\delta u} = m, \quad \frac{\delta \ell}{\delta a} = b, \quad 
\rmd g\,g^{-1} = u\,dt + \sum_i \frac{\delta h_i}{\delta m}\circ dW^i_t, \quad  \rmd b = \frac{\delta \ell}{\delta a}\,dt - \sum_i \frac{\delta h_i}{\delta a}\circ dW_t^i, \\
&\rmd m = - \ad^*_{\rmd g\,g^{-1}}m +  \rmd b \diamond a, \quad \rmd a = -\mathsterling_{\rmd g\,g^{-1}}a.
\end{split} \label{eq:EP eq parts}
\end{align}
Thus the SALT EP equations are found to be
\begin{align}
\rmd \frac{\delta \ell}{\delta u} = -\, \ad^*_{\rmd g\,g^{-1}}\frac{\delta \ell}{\delta u} + \left(\frac{\delta \ell}{\delta a}\,dt - \sum_i \frac{\delta h_i}{\delta a}\circ dW_t^i\right) \diamond a
\,, \quad 
\rmd a = -\mathsterling_{\rmd g\,g^{-1}} a\,,
\label{eq:StochEPeqns}
\end{align}
where the definition of $\rmd g\,g^{-1}$ are taken from \eqref{eq:EP eq parts}. For a similar, but more rigorous approach to the derivation of these equations, see \cite{SC2020}.

\paragraph{Lie-Poisson (LP) Hamiltonian formulation.}
Using the Legendre transform, $h(m,a) = \SCP{m}{u} - \ell(u,a)$ and taking variations yields
\begin{align}
\frac{\delta \ell}{\delta u} = m\,, \quad \frac{\delta h}{\delta m} = u\,, \quad \frac{\delta h}{\delta a} = -\frac{\delta \ell}{\delta a}\,,
\end{align}
and the corresponding SALT Lie-Poisson (LP) Hamiltonian equations obtained after a Legendre transform are
\begin{align}
\begin{split}
\rmd m = - \ad^*_{\rmd x_t}m - \left(\frac{\delta h}{\delta a}\,dt 
+ \sum_i \frac{\delta h_i}{\delta a}\circ dW_t^i\right) \diamond a
\,, \quad \rmd a = -\mathsterling_{\rmd x_t} a\,,
\end{split} \label{eq:SLP eq SALT}
\end{align}
where the Lagrangian path $\rmd x_t = \rmd g\,g^{-1}$ Legendre-transforms to the Hamiltonian side as
\begin{align}
\rmd x_t = \frac{\delta h}{\delta m}\,dt + \sum_i \frac{\delta h_i}{\delta m}\circ dW^i_t
\,.
\label{LagPath-SALT}
\end{align} 

By using the $\ad^*$ operator for semidirect product Lie algebras defined in appendix \ref{app: semidirect ad^*}, the LP equations \eqref{eq:SLP eq SALT} can be written equivalently in the following compact form,
\begin{align*}
\rmd (m,a) = -\ad^*_{(\rmd x_t,  \rmd b)}(m,a), \quad \text{where }  \rmd b = -\frac{\delta h}{\delta a}\,dt - \sum_i \frac{\delta h_i}{\delta a}\circ dW_t^i.
\end{align*}

\begin{remark}[Stochastic reduced Hamiltonian phase space variational principle for SALT]
	The SALT equations \eqref{eq:SLP eq SALT} can be derived from a stochastic reduced Hamilton phase-space variational principle, namely,
	\begin{align*}
	0=\delta S = \delta \int^b_a \SCP{m}{\rmd g\, g^{-1}} + \SCP{\rmd b}{a_0g^{-1} - a} - h(m, a)\,dt - \sum_i h_i(m,a)\circ dW^i_t,
	\end{align*}
	where the variations $\delta m, \delta g, \delta a$ and $\delta( b)$ are taken to be arbitrary.     
\end{remark}

The SALT Hamiltonian equations in \eqref{eq:SLP eq SALT} can be arranged into the Lie-Poisson (LP)  operator form
\begin{align}
\rmd \begin{bmatrix}m \\ a \end{bmatrix} = -
\begin{bmatrix}
\ad^*_{\fbox{}} m & \fbox{}\diamond a \\
\mathsterling_{\fbox{}}a & 0 
\end{bmatrix}
\begin{bmatrix}
\frac{\delta h}{\delta m}\,dt + \sum_i \frac{\delta h_i}{\delta m} \circ dW^i_t 
\\
\frac{\delta h}{\delta a}\,dt + \sum_i \frac{\delta h_i}{\delta a} \circ dW^i_t 
\end{bmatrix}\,,
\label{eq:SALT-LPB}
\end{align}
The Lie-Poisson operator in \eqref{eq:SALT-LPB} preserves the Casimirs of its deterministic counterpart, since the Poisson structures remains the same. However, the Hamiltonian is now a semimartingale, so energy depends explicitly on time and, hence, is no longer preserved.

\paragraph{Kelvin-Noether theorem.}
The SALT Lie-Poisson (LP) Hamiltonian equations in \eqref{eq:SLP eq SALT} also possess a Kelvin-Noether theorem. 
To understand this statement, consider the following $G$-equivariant map
$\mathcal{K}:C\times V^* \rightarrow \mathfrak{X}^{**}$ as explained in \cite{HMR1998}, 
\[
\SCP{\mathcal{K}(gc, ag^{-1})}{\Ad^*_{g^{-1}}v} = \SCP{\mathcal{K}(c, a)}{v}
\,,\quad  \hbox{for all}\quad g \in G\,,
\]
for a manifold $C$ on which $G$ acts from the left. For fluid dynamics, $C$ is the space of loops the fluid domain $\mathcal{D}$ and the map $\mathcal{K}$ is the circulation around the loop. More specifically, for all $\alpha \in \Lambda^1$,
\begin{align}
    \SCP{\mathcal{K}(c, a)}{\alpha} := \oint_{c} \alpha,
\end{align}
    \begin{theorem}[SALT Kelvin-Noether theorem] \label{SKN theorem}
	Given solutions $m(t), a(t)$ satisfying the SALT LP equations \eqref{eq:SLP eq SALT} and fixed $c_0 \in C$, the associated Kelvin-Noether quantity $\SCP{\mathcal{K}(g(t)c_0, a(t))}{m(t)}$ satisfies the following stochastic Kelvin-Noether relation.
	\begin{align}
	\rmd \SCP{\mathcal{K}(g(t)c_0, a(t))}{m(t)} = \SCP{\mathcal{K}(g(t)c_0, a(t))}{\left(- \frac{\delta h}{\delta a}\,dt - \sum_i \frac{\delta h_i}{\delta a}\circ dW_t^i \right)\diamond a}
	\end{align}
	where we identify $\rmd g\,g^{-1} =: \rmd x_t$. 
\end{theorem}
\begin{remark}
    The proof of theorem \ref{SKN theorem} is in appendix \ref{app: KIW}. The fluid mechanics counterpart of the Kelvin-Noether theorem is expressed as
	\begin{align*}
	\rmd \oint_{c_t} \frac{m}{D} = \oint_{c_t} \frac{1}{D}\left(- \frac{\delta h}{\delta a}\,dt - \sum_i \frac{\delta h_i}{\delta a}\circ dW_t^i \right)\diamond a\,, 
	\end{align*} 
	where the material loop $c_t=g(t)c_0$ moves with stochastic velocity $\rmd g\, g^{-1}=\rmd x_t$ in equation \eqref{LagPath-SALT} and the quantity $D$ is the mass density of the fluid which is also advected as $D(t)=g_{t\,*}D_0$. That is, the mass density $D$ satisfies the stochastic continuity equation $\rmd D + \mathcal{L}_{\rmd x_t}D=0$.
\end{remark}

\begin{remark}[SALT Hamiltonians]
	In \cite{Holm2015}, the SEP equations are derived from a stochastic Clebsch variational principle where the advection of phase space Lagrangian variables are through a stochastic vector field $\rmd x_t = u\,dt + \sum_i \xi_i(x)\circ dW_t^i$. Compared to the stochastic vector field defined via $\rmd g\,g^{-1}$ in \eqref{eq:EP eq parts}, they coincide with those in \cite{Holm2015} when the noise Hamiltonians $h_i$ are linear in $m$, i.e. $h_i(m,a) = \SCP{\xi_i(x)}{m}$. The SEP equations then becomes
	\begin{align*}
	\rmd \frac{\delta \ell}{\delta u} = -\, \ad^*_{\rmd x_t}\frac{\delta \ell}{\delta u} + \frac{\delta \ell}{\delta a} \diamond a \,dt \,, \quad 
	\rmd a = -\mathsterling_{\rmd x_t} a\,, \quad \rmd x_t = u\,dt + \sum_i \xi_i(x)\circ dW_t^i\,.
	\end{align*}
	For data assimilation purposes, the choice $h_i(m,a) = \sum_i\langle m,\xi_i(x) \rangle$ was made in \cite{CCHOS18,CCHOS18a}.
\end{remark}

\subsection{SALT Kelvin theorem via the Kunita-It\^o-Wentzell theorem}\label{app: KIW}

A pair of results in Kunita \cite{kunita1981,kunita1984} provided the key to working with stochastic advection by Lie transport (SALT) in ideal fluid dynamics. In particular, if we choose the diffeomorphism  $\phi_t$ as the stochastic process obtained by homogenisation in \cite{CGH17}
\[
{\rmd }\phi_t(x)
:=
u(\phi_t(x),t)dt + \sum_{i}\xi_{i}(\phi_t(x))\circ dW_{t}^{i}
\,,
\]
then the Kunita It\^o-Wentzell change of variables formula discussed below leads to the following differential form leads to the stochastic advection law,
\[
{\rmd } \big(\phi_t^*K(t,x)\big) = \phi_t^* \Big({\rmd } K(t,x) 
+ \mathcal{L}_{\rmd \phi_t(x)} K(t,x)\Big)= 0
\,,\quad\hbox{a.s.}
\]
where $\mathcal L_{\rmd \phi_t(x)}$ is the \emph{Lie derivative} by the vector field $\rmd \phi_t(x)$ whose time integral 
$\int_0^t\rmd \phi_s(x) = \phi_t (x) - \phi_0(x) $ generates the semimartingale flow $\phi_t$ acting  on semimartingale $k$-form, $\rmd K(t,x) = G(t,x) d t + H(t,x) \circ d W_t$.  One may recall that the Lie derivative $\mathcal{L}_{\rmd \phi_t} K$ has both a dynamic and a geometric definition,
\[
\mathcal{L}_{\rmd \phi_t} K 
= \lim_{\Delta s\to0} \frac{1}{\Delta s} (\phi_{\Delta s}^* K - K) 
= \rmd \phi_t \contract dK + d (\rmd \phi_t \contract K) 
\]
in which the latter formula is attributed to Cartan. 

Here is a simplified statement of the theorem for applying the Kunita It\^o-Wentzell change of variables formula to stochastic advection of differential $k$-forms proved in \cite{dLHLT2020}, based on \cite{kunita1981,kunita1984}. See also \cite{LJW84}.

\begin{theorem}[Kunita-It\^o-Wentzell (KIW) formula for $k$-forms]
	Consider a sufficiently smooth $k$-form $K(t,x)$ in space which is a semimartingale in time
	\begin{align} \label{spde-compact}
	\rmd K(t,x) = G(t,x) dt + \sum_{i=1}^M H_i(t,x) \circ dW_t^i,
	\end{align}
	where $W_t^i$ are i.i.d. Brownian motions. Let $\phi_t$ be a sufficiently smooth flow satisfying the SDE 
	\[
	\rmd \phi_{t}(x) = b(t,\phi_{t}(x)) d t + \sum_{i=1}^N \xi_i(t,\phi_{t}(x)) \circ d B_t^i
	\,,\]
	in which  $B_t^i$ are i.i.d. Brownian motions. Then the pull-back $\phi_t^* K$ satisfies the formula
	\begin{align}
	\begin{split}
	\rmd \,(\phi_t^* K)(t,x) = &\phi_t^* G(t,x) d t + \sum_{i=1}^M \phi_t^* H_i(t,x) \circ d W_t^i \\
	&+  \phi_t^*\mathcal L_b K(t,x) d t  + \sum_{i=1}^N \phi_t^* \mathcal L_{\xi_i} K(t,x) \circ d B_t^i.
	\label{KIWkformsimplified}
	\end{split}
	\end{align}
\end{theorem}
Formulas \eqref{spde-compact} and \eqref{KIWkformsimplified} are compact forms of equations in \cite{dLHLT2020} which are written in integral notation to make the stochastic processes more explicit.

To understand the distinction between integral and differential notation for SPEs, one may begin by writing  the stochastic `fundamental theorem of calculus' as 
\[
\phi_t^* K(t,x) - \phi_0^*  K(0,x)
:= K(t,\phi_t(x)) -  K(0,x) = \int^t_0 {\rmd } \,(\phi_s^*K_s) 
\,.\]
In the integral notation, the Kunita-It\^o-Wenzell (KIW) formula is written as
\[
\int^t_0 {\rmd } \,(\phi_s^*K_s)  = \int^t_0 \phi_s^* \big( {\rmd } K(s,x)   
+ \mathcal L_{\rmd \phi_s(x)} K(s,x)\big) { } \,.
\]\
So, in the differential notation the KIW formula `transfers' to the equivalent differential form
\[
{\rmd } \big(\phi_t^*K(t,x)\big) = \phi_t^* \Big({\rmd } K(t,x) 
+ \mathcal{L}_{\rmd \phi_t(x)} K(t,x)\Big)
\,,\quad\hbox{a.s.}
\]
\begin{remark}
	In applications, one sometimes expresses equation \eqref{KIWkformsimplified} using the differential notation
	\begin{align*}
	\rmd \,(\phi_t^* K)(t,x) = \phi_t^*\left(\rmd K + \mathcal L_{\rmd x_t} K\right)(t,x),
	\end{align*}
	where $\rmd x_t$ is the stochastic vector field $\rmd x_t(x) = b(t,x) \rmd t + \sum_{i=1}^N \xi_i(t,x) \circ \rmd B_t^i$. Importantly for fluid dynamics, this formula is also valid when $K$ is a vector field rather than a $k$-form.
\end{remark}

\paragraph{Stochastic Kelvin circulation theorem for the SALT theory.}
Having understood the differential notation for stochastic integrals, now we may assemble the stochastic fluid equations via the stochastic Kelvin circulation theorem for the SALT theory \cite{Holm2015}. For this purpose, we shall make the argument that the stochastic Kelvin circulation theorem is fundamentally a stochastic form of Newton's law of motion,
\[
{\color{red}\bf d }\oint_{c({\rmd }\phi_t)} \!\!\!\mathbf{v}\cdot d \mathbf{x}
= \oint_{c({\rmd }\phi_t) }
\underbrace{\ 
	({\color{red}\bf d } + \mathcal{L}_{{\rmd }\phi_t}) (\mathbf{v} \cdot d \mathbf{x})\ 
}_{\sf KIW\ formula}
= \oint_{c({\rmd }\phi_t) }
\underbrace{\ 
	\mathbf{f} \cdot d \mathbf{x}\  
}_{\sf Newton's\ Law}
.
\]
This formula corresponds to the \emph{motion equation} derived from Hamilton's principle 
\[
\Big({\color{red}\bf d } + \mathcal{L}_{{\rmd }\phi_t}\Big) 
\Big(\frac1D\frac{\delta \ell}{\delta \mathbf{u}}  \cdot d \mathbf{x} \Big)
= \mathbf{f} \cdot d \mathbf{x} \,,
\]
along with the law of  \emph{advection of mass} expressed in KIW form
\[
\Big({\color{red}\bf d } + \mathcal{L}_{{\rmd }\phi_t}\Big) \big(D d^3x\big) = 0\,,
\]
where the \emph{flow velocity} is given by the stochastic vector field
\[
{\rmd }\phi_t(x)
:=
u(\phi_t(x),t)dt+{\sum_{i}\xi_{i}(\phi_t(x))\circ dW_{t}^{i}}
\,.
\]
Thus, the stochastic Kelvin's circulation theorem for SALT simply describes the rate of change of momentum of a \emph{stochastically moving} material loop.

\section{It\^o form of the SFLP equation}\label{app: Ito form}
\begin{proposition}
	The It\^o form of the SFLP equation \eqref{SAPS-LP-eq} 
	\begin{align*}
	\rmd m + \ad^*_\frac{\delta h}{\delta m} m\,dt + \sum_i \ad^*_\frac{\delta h}{\delta m} f_i \circ dW^i_t = 0\,, 
	\end{align*}
	is given by 
	\begin{align}
	\rmd m + \ad^*_{\frac{\delta h}{\delta m}}m\,dt + \frac{1}{2}\sum_i \ad^*_{\sigma_i}f_i\,dt + \sum_i \ad^*_{\frac{\delta h}{\delta m}}f_i dW_t^i = 0, \quad  \sigma_i = \qp{-\ad^*_{\frac{\delta h}{\delta m}}f_i}{\frac{\delta^2 h}{\delta m^2}},
	\end{align}
	where the brackets $\qp{\cdot}{\cdot}$ in the definition of $\sigma_i$ denotes contraction, not $L^2$ pairing.
\end{proposition}
\begin{proof}
	This can be shown via direct computation. We ignore the drift term and suppress the indices on constants $f_i$ for ease of notation, choosing an arbitrary $\phi = \phi(x) \in \mathfrak{X}$, we have the Stratonovich stochastic equation
	\begin{align}
	\rmd \scp{m}{\phi} = \Scp{-\ad^*_{\frac{\delta h}{\delta m}} f}{\phi} \circ dW_t = \SCP{\ad^*_{\phi} f}{\frac{\delta h}{\delta m}} \circ dW_t
	\end{align}
	The corresponding It\^o form is then
	\begin{align}
	\begin{split}
	\rmd \scp{\phi}{m} &= \SCP{\ad^*_{\phi} f}{\frac{\delta h}{\delta m}}\, dW_t + \frac{1}{2}\qv{\rmd\SCP{\frac{\delta h}{\delta m}}{\ad^*_\phi f}}{dW_t}\\
	&= \SCP{\ad^*_{\phi} f}{\frac{\delta h}{\delta m}}\, dW_t + \frac{1}{2}\qv{\SCP{\rmd\frac{\delta h}{\delta m}}{\ad^*_\phi f}}{dW_t}\\ 
	&= \SCP{\ad^*_{\phi} f}{\frac{\delta h}{\delta m}}\, dW_t + \frac{1}{2}\qv{\SCP{\qp{-\frac{\delta^2 h}{\delta m^2}}{\ad^*_{\frac{\delta h}{\delta m}} f}\,dW_t}{\ad^*_\phi f}}{dW_t}\\ 
	&= -\SCP{\ad^*_{\frac{\delta h}{\delta m}} f}{\phi}\, dW_t - \frac{1}{2}\SCP{\ad^*_\sigma f}{\phi}\qv{dW_t}{dW_t}\\
	\end{split}
	\end{align}
	Since $\phi$ is arbitrary, the It\^o form of \eqref{SAPS-LP-eq} is 
	\begin{align}
	\rmd m + \ad^*_{\frac{\delta h}{\delta m}}m\,dt + \sum_i \ad^*_{\frac{\delta h}{\delta m}}f_i\, dW_t^i + \frac{1}{2}\sum_{i,j} \ad^*_{\sigma_j}f_i \qv{dW_t^i}{dW_t^j} = 0.
	\end{align}
	For Brownian motion, the quadratic variation term simplifies to $\qv{dW_t^i}{dW_t^j} = \delta^{ij}dt$ which recovers It\^o form as presented.
\end{proof}
\begin{proposition}\label{Ito form semidirect product}
    The It\^o form of the SFLP equation with advected quantity \eqref{ExplicitSDPergEqns} 
    \begin{align}
    \begin{split}
        &\rmd m + \ad^*_\frac{\delta h}{\delta m}m\,dt + \frac{\delta h}{\delta a}\diamond a\,dt + \sum_i \ad^*_\frac{\delta h}{\delta m}f^i \circ dW^i_t + \frac{\delta h}{\delta a}\diamond g^i \circ dW^i_t = 0\\
        &\rmd a + \mathcal{L}_\frac{\delta h}{\delta m}a\,dt + \sum_i \mathcal{L}_\frac{\delta h}{\delta m} g^i \circ dW^i_t = 0
    \end{split}
    \end{align}
    is given by
\begin{align}
    \begin{split}
        &\rmd m + \ad^*_\frac{\delta h}{\delta m}m\,dt + \frac{\delta h}{\delta a}\diamond a\,dt + \sum_i \ad^*_\frac{\delta h}{\delta m}f^i \, dW^i_t + \frac{\delta h}{\delta a}\diamond g^i \, dW^i_t + \frac{1}{2}\sum_i (\ad^*_{\sigma_i}f^i + \theta_i\diamond g^i)\,dt = 0\\
        &\rmd a + \mathcal{L}_\frac{\delta h}{\delta m}a\,dt + \sum_i \mathcal{L}_\frac{\delta h}{\delta m} g^i \, dW^i_t + \frac{1}{2}\sum_i \mathcal{L}_{\sigma_i}g^i \,dt = 0
    \end{split}
\end{align}
where $\sigma_i$ and $\theta_i$ are given as
\begin{align*}
    \sigma_i := \qp{\frac{\delta^2 h}{\delta m^2}}{-\ad^*_{\frac{\delta h}{\delta m}}f^m_i - \frac{\delta h}{\delta a}\diamond f^a_i}
    \quad \hbox{and}\quad
    \theta_i := \qp{\frac{\delta^2 h}{\delta a^2}}{-\mathcal{L}_{\frac{\delta h}{\delta m}}f^a_i}.
    \end{align*}
\end{proposition}
\begin{proof}
Consider similarly by ignoring drift terms and suppress the indices on the $f^i \in \mathfrak{X}^*$, $g^i\in V^*$, $\sigma_i\in \mathfrak{X}$ and $\theta_i \in V$. Consider constants $\phi = \phi(x) \in \mathfrak{X}$ and $\psi = \psi(x) \in V$, the following Stratonovich stochastic equations hold
\begin{align}
    \rmd \SCP{m}{\phi} = -\SCP{\ad^*_\frac{\delta h}{\delta m}f + \frac{\delta h}{\delta a}\diamond g}{\phi}\circ dW_t, \quad \rmd \SCP{a}{\psi} = -\SCP{\mathcal{L}_\frac{\delta h}{\delta m}g}{\psi}\circ dW_t = 0
\end{align}
The It\^o form thus satisfies
\begin{align*}
    \rmd \SCP{m}{\phi} &= -\SCP{\ad^*_\frac{\delta h}{\delta m}f + \frac{\delta h}{\delta a}\diamond g}{\phi}dW_t - \frac{1}{2}\qv{\rmd \SCP{\ad^*_\frac{\delta h}{\delta m}f + \frac{\delta h}{\delta a}\diamond g}{\phi}}{dW_t}\\
    & = -\SCP{\ad^*_\frac{\delta h}{\delta m}f + \frac{\delta h}{\delta a}\diamond g}{\phi}dW_t + \frac{1}{2}\qv{\rmd \SCP{\ad^*_\phi f}{\frac{\delta h}{\delta m}} + \rmd \SCP{\mathcal{L}_\phi g}{\frac{\delta h}{\delta a}}}{dW_t}\\
    & = -\SCP{\ad^*_\frac{\delta h}{\delta m}f + \frac{\delta h}{\delta a}\diamond g}{\phi}dW_t + \frac{1}{2}\qv{\SCP{\ad^*_\phi f}{\rmd \frac{\delta h}{\delta m}} + \SCP{\mathcal{L}_\phi g}{\rmd \frac{\delta h}{\delta a}}}{dW_t}\\
    & = -\SCP{\ad^*_\frac{\delta h}{\delta m}f + \frac{\delta h}{\delta a}\diamond g}{\phi}dW_t + \frac{1}{2}\qv{\SCP{\ad^*_\phi f}{\sigma_i\, dW_t} + \SCP{\mathcal{L}_\phi g}{\theta_i\, dW_t}}{dW_t}\\
    & = -\SCP{\ad^*_\frac{\delta h}{\delta m}f + \frac{\delta h}{\delta a}\diamond g}{\phi}dW_t - \frac{1}{2}\SCP{\ad^*_\sigma f + \mathcal{L}_\theta g}{\phi}\,\qv{dW_t}{dW_t}\\
    \rmd \SCP{a}{\psi} &= -\SCP{\mathcal{L}_\frac{\delta h}{\delta m}g}{\psi}\,dW_
    t - \frac{1}{2}\qv{\rmd \SCP{\mathcal{L}_\frac{\delta h}{\delta m}g}{\psi}}{dW_t}\\
    & = -\SCP{\mathcal{L}_\frac{\delta h}{\delta m}g}{\psi}\,dW_
    t + \frac{1}{2}\qv{\rmd \SCP{\psi\diamond g}{\frac{\delta h}{\delta m}}}{dW_t}\\
    & = -\SCP{\mathcal{L}_\frac{\delta h}{\delta m}g}{\psi}\,dW_
    t - \frac{1}{2}\qv{\SCP{\psi\diamond g}{\sigma\, dW_t}}{dW_t}\\
    & = -\SCP{\mathcal{L}_\frac{\delta h}{\delta m}g}{\psi}\,dW_
    t - \frac{1}{2}\SCP{\mathcal{L}_\sigma g}{\psi}\,\qv{dW_t}{dW_t}
\end{align*}
Since $\phi$ and $\psi$ are arbitrary, the It\^o form of $\rmd m$ and $\rmd a$ with drift are
\begin{align}
\begin{split}
&\rmd m + \ad^*_\frac{\delta h}{\delta m}m\,dt + \frac{\delta h}{\delta a}\diamond a \,dt + \sum_i \left(\ad^*_\frac{\delta h}{\delta m}f^i + \frac{\delta h}{\delta a}\diamond g^i\right)\,dW^i_t + \frac{1}{2}\sum_{i,j}(\ad^*_{\sigma_j}f^i + \theta_j\diamond g^i)\qv{dW^i_t}{dW^j_t} = 0\\
&\rmd a + \mathcal{L}_\frac{\delta h}{\delta m} a\,dt + \sum_{i,j}\mathcal{L}_{\sigma_j}g^i\,\qv{dW^i_t}{dW^j_t}= 0   
\end{split}
\end{align}
For Brownian motion, the quadratic variation term simplifies to $\qv{dW_t^i}{dW_t^j} = \delta^{ij}dt$ which completes the proof.
\end{proof}

\end{document}